\documentclass[a4paper,times,10pt,onecolumn,review,3p]{elsarticle}  
\biboptions{sort&compress,numbers}           % only for elsevier class
\usepackage{amsfonts}                                           
\usepackage{amssymb}                                           
\usepackage{amsmath,mathrsfs}						                   
\usepackage{amsxtra}                                            
\usepackage{amsthm}						                      
\usepackage{accents}
\usepackage{times}
\usepackage{xr}
\usepackage[english]{babel}
\usepackage{booktabs}                                               % -> für anzeige von deutschen Umlauten
\usepackage{textcomp}                                               % for \textdegree=° im utf8 modus
\usepackage[dvips,pdfborder={0 0 0},hypertexnames=false,extension=pdf]{hyperref} % version für submission
\usepackage{graphicx}                                        % Eps Grafiken für PostScript Drucker
\usepackage{epsfig}
\usepackage[utf8]{inputenc}                                         % -> damit man sie auch gleich im Code eingeben kann
\usepackage[T1]{fontenc}
\usepackage{subfigure}
\usepackage{xcolor}
\usepackage{verbatim}
\usepackage[all]{xy}
\usepackage{enumerate}                                              % to switch easy the enumerate style
\hypersetup {
  pdftitle          = {Approximation of Löwdin Orthogonalization to a Spectrally Efficient Orthogonal 
                       Overlapping PPM Design for UWB Impulse Radio},
  pdfauthor 	    = {Philipp Walk, Peter Jung},
  pdfsubject 	    = {},
  pdfkeywords 	    = {},
  bookmarksnumbered = true,   % haut auf keinen fall hin mit amsrefs colorlinks	    = false}
}

\newcommand{\del}{\ensuremath{\delta}}
\newcommand{\Del}{\ensuremath{\Delta}}

\newcommand{\gam}{\ensuremath{\gamma}}

\newcommand{\lam}{\ensuremath{\lambda}}
\newcommand{\Lam}{\ensuremath{\Lambda}}
\newcommand{\ome}{\ensuremath{\omega}}

\newcommand{\sig}{\ensuremath{\sigma}}

\newcommand{\tp}{\ensuremath{\tilde{p}}}

\newcommand{\hp}{\ensuremath{\hat{p}}}
\newcommand{\hvc}{\ensuremath{\hat{\vc}}}

\newcommand{\Laplacep}{\ensuremath{(\mathscr{L}{p})}}

\newcommand{\hq}{\ensuremath{\hat{q}}}

\newcommand{\vp}{\ensuremath{\mathbf p}}

\newcommand{\vpM}{\ensuremath{\mathbf p}^M}

\newcommand{\po}{\ensuremath{p^\circ}}
\newcommand{\pTo}{\ensuremath{p^{T,\circ}}}
\newcommand{\pToM}{\ensuremath{p^{T,\circ,M}}}
\newcommand{\poD}{\ensuremath{p^{\Delta,\circ}}}
\newcommand{\tpo}{\ensuremath{\tilde{p}^\circ}}
\newcommand{\tpToM}{\ensuremath{\tilde{p}^{T,\circ,M}}}

\newcommand{\htpo}{\ensuremath{\hat{\tilde{p}}^\circ}}

\newcommand{\qoD}{\ensuremath{q^{\Delta,\circ}}}

\newcommand{\poo}{\ensuremath{p^\circ_0}}
\newcommand{\poM}{\ensuremath{p^{\circ,M}}}

\newcommand{\hpo}{\ensuremath{\hat{p}^\circ}}
\newcommand{\hpoD}{\ensuremath{\hat{p}^{\Delta,\circ}}}

\newcommand{\hpoMm}{\ensuremath{\hat{p}^{\circ,M}_m}}

\newcommand{\hqoD}{\ensuremath{\hat{q}^{\Delta,\circ}}}

\newcommand{\tpoM}{{\ensuremath{{\tilde{p}}^{\circ,M}}}}
\newcommand{\tpoMm}{{\ensuremath{{\tilde{p}}^{\circ,M}_m}}}
\newcommand{\tpoMk}{{\ensuremath{{\tilde{p}}^{\circ,M}_k}}}
\newcommand{\tpoKk}{{\ensuremath{{\tilde{p}}^{\circ,K}_k}}}
\newcommand{\tpoKo}{{\ensuremath{{\tilde{p}}^{\circ,K}_0}}}

\newcommand{\tpoMnM}{{\ensuremath{{\tilde{p}}^{\circ,M}_{m\!-\!M}}}}

\newcommand{\poKo}{{\ensuremath{p}^{\circ,K}_0}}

\newcommand{\poMm}{{\ensuremath{p}^{\circ,M}_m}}

\newcommand{\tvpoM}{\ensuremath{{\tilde{\mathbf p}}^{\circ,M}}}

\newcommand{\Noise}{{\ensuremath{\mathcal{N}_0}}}

\newcommand{\Hil}{{\ensuremath{\mathcal H}}}

\newcommand{\En}{{\ensuremath{\mathcal {E}}}}

\newcommand{\SI}{\ensuremath{ V }}

\newcommand{\C}{{\ensuremath{\mathbb C}}}

\newcommand{\R}{{\ensuremath{\mathbb R}}}
\newcommand{\N}{{\ensuremath{\mathbb N}}}

\newcommand{\Z}{{\ensuremath{\mathbb Z}}}

\newcommand{\M}{{\ensuremath{\mathcal{M}}}}

\newcommand{\mC}{{\ensuremath{\mathbf C}}}
\newcommand{\mD}{{\ensuremath{\mathbf D}}}

\newcommand{\Gram}{{\ensuremath{\mathbf G}}}

\newcommand{\GramM}{{\ensuremath{{\mathbf G}_M}}}

\newcommand{\tGramM}{{\ensuremath{{\mathbf {\tilde{G}}}_M}}}

\newcommand{\GramMinv}{{\ensuremath{\mathbf G}_M^{-\frac{1}{2}}}}

\newcommand{\Graminv}{{\ensuremath{\mathbf G}^{-\frac{1}{2}}}}

\newcommand{\GramMi}{{\ensuremath{\mathbf G}^{-1}_M}}

\newcommand{\tGramMinv}{{\ensuremath{\mathbf {\tilde{G}}}_M^{-\frac{1}{2}}}}

\newcommand{\FmatrixM}{{\ensuremath{\mathbf F}_M}}

\newcommand{\FmatrixMa}{{\ensuremath{\mathbf F}_M^*}}

\newcommand{\tDiagM}{{\ensuremath{\tilde{{\mathbf D}}_M}}}

\newcommand{\tlam}{{\ensuremath{\tilde{\lam}}}}

\newcommand{\tDiagMinv}{{\ensuremath{\mathbf {\tilde{D}}}_M^{-\frac{1}{2}}}}

\newcommand{\Zak}{{\ensuremath{\mathbf Z}}}

\newcommand{\eins}{{\ensuremath{\mathbf 1}}}

\newcommand{\LRA}{\ensuremath{\Leftrightarrow} }

\newcommand{\vg}{{\ensuremath{\mathbf g}}}
\newcommand{\vgL}{{\ensuremath{\mathbf g}}}

\newcommand{\vgtDg}{{\ensuremath{{{\tilde{\mathbf g}}}_{\Delta,{\vg}}}}}
\newcommand{\hvgtDg}{{\ensuremath{{\hat{\tilde{\mathbf g}}}_{\Delta,{\vg}}}}}

\newcommand{\hvgtDzweig}{{\ensuremath{{\hat{\tilde{\mathbf g}}}_{2,{\vg}}}}}
\newcommand{\hvgL}{{\ensuremath{\hat{\mathbf g}}}}
\newcommand{\hvg}{{\ensuremath{\hat{\mathbf g}}}}

\newcommand{\va}{{\ensuremath{\mathbf a}}}
\newcommand{\vc}{{\ensuremath{\mathbf c}}}

\newcommand{\vd}{{\ensuremath{\mathbf d}}}

\newcommand{\vphi}{\ensuremath{\boldsymbol{ \phi}}}

\newcommand{\vgm}{\ensuremath{{\mathbf g}_m}}

\newcommand{\vP}{\ensuremath{\mathbf P }}                         % Vektorwertiges Ma{\ss}
\newcommand{\vPM}{\ensuremath{\mathbf P\!}_M}                         % Vektorwertiges Ma{\ss}
\newcommand{\vr}{\ensuremath{\mathbf r}_{\vg}}                         % Vektorwertige Funktion
\newcommand{\vh}{\ensuremath{\mathbf b }}                         % Vektorwertige Funktion
\newcommand{\vbb}{\ensuremath{\mathbf b }}                         % Vektorwertige Funktion
\newcommand{\hvr}{{\ensuremath{\hat{\mathbf r}_{\vg}}}}
\newcommand{\hvh}{{\ensuremath{\hat{\mathbf b}}}}
\newcommand{\ttau}{\ensuremath{\tilde\tau }}                         % Vektorwertige Funktion
\newcommand{\thmref}[1]{Satz~\ref{#1}}

\newcommand{\lemref}[1]{Lemma~\ref{#1}}

\newcommand{\secref}[1]{Section~\ref{#1}}

\newcommand{\figref}[1]{Abbildung~\ref{#1}}

\newcommand{\noi}{\noindent}

\ifx\lemma\undefined
\newtheorem{lemma}{Lemma}
\fi
\ifx\theorem\undefined
\newtheorem{theorem}{Theorem}[section]         %%%%%%%%%%%%%%%%%%%%%%%%%%%%%%
\fi
\ifx\definition\undefined
\newtheorem{definition}[theorem]{Definition}         %%%%%%%%%%%%%%%%%%%%%%%%%%%%%%
\fi
\ifx \@definition \@empty
\newtheorem{definition}[theorem]{Definition}   % the section where they     %
\fi
\ifx \@remark \@empty
\fi
\ifx \@proposition \@empty
\fi
\newtheorem{corrolary}{Corrolary} %%%%%%%%%%%%%%%%%%%%%%%%%%%%%% 
{\par\noindent{\em Beweis\/}.}%
{\hspace*{\fill}{\qed}\vspace{1ex}\par}
{\par\noindent{\em Proof\/}.}%
{\par}
\newenvironment{bemerkung}%
{\par\vspace{1.5ex}\noindent{\em Remark\/}.}
{\par\vspace{1.5ex}}
{\par\vspace{1.5ex}\noindent{\em Example\/ }}
{\par\vspace{1.5ex}}
{\noi\vspace{0.5ex}\small}
{\vspace{0.5ex}\par\normalsize}

\ifx\example\@undefined
\newenvironment{example}%
{\par\vspace{2ex}\small\noindent{\em Example\/}.}
{\vspace{2ex}\par\normalsize}
\else
\fi

{\renewcommand{\labelenumi}{(\roman{enumi})}\begin{list}{\labelenumi}
{\usecounter{enumi}\setlength{\labelwidth}{1.5cm}\setlength{\topsep}{0.3cm}\setlength{\itemsep}{-3pt}}}
{\end{list}}
{\renewcommand{\labelenumi}{(\arabic{enumi})}\begin{list}{\labelenumi}
{\usecounter{enumi}\setlength{\labelwidth}{1.5cm}\setlength{\topsep}{0.3cm}\setlength{\itemsep}{-3pt}}}
{\end{list}}
{\renewcommand{\labelenumi}{$\bullet$}\begin{list}{\labelenumi}
{\setlength{\labelwidth}{1.5cm}\setlength{\topsep}{0.3cm}\setlength{\itemsep}{-2pt}}}
{\end{list}}

\makeatletter %%%% wichtig damit die @ befehle gehen
\newcommand{\sprod}[2]{\ensuremath{%
\setbox0=\hbox{\ensuremath{#2}}
\dimen@\ht0
\advance\dimen@ by \dp0
\left(\left.#1\rule[-\dp0]{0pt}{\dimen@}\right|#2\right)}}
\makeatother %%%%%%% ende @ befehle

\makeatletter
\newcommand{\set}[2]{\ensuremath{%
\setbox0=\hbox{\ensuremath{#2}}
\dimen@\ht0
\advance\dimen@ by \dp0
\left\{\left.#1\rule[-\dp0]{0pt}{\dimen@}\;\right|\;#2\right\} }}
\makeatother

\newcommand{\Set}[1]{\ensuremath{\left\{ #1 \right\}}}

\DeclareMathOperator{\spann}{span}

\newcommand{\convdis}{{\ensuremath{\ *^{\prime}\, }}}
\newcommand{\convdisT}{{\ensuremath{\ *^{\prime}_T\, }}}
\newcommand{\convdisTo}{{\ensuremath{\ *^{\prime}_{T_0}\, }}}
\newcommand{\convdisD}{{\ensuremath{\ *^{\prime}_\Delta\, }}}

\DeclareMathOperator{\esup}{ess\, sup}

\newcommand{\esssup}[2]{\ensuremath{ \underset{#1}{\esup}\left\{#2\right\}}}

\newcommand{\Betrag}[1]{\ensuremath{ \left|#1\right| }}
\newcommand{\Abs}[1]{\ensuremath{ \left|#1\right| }}
\newcommand{\Norm}[1]{\ensuremath{ \left\|#1\right\| }}

\newcommand{\SP}[2]{\ensuremath{ \langle#1,#2\rangle }}
\newcommand{\SPH}[2]{\ensuremath{ \langle#1,#2\rangle }_{_{\Hil}}}

\newcommand{\Expect}[1]{{\ensuremath{\mathbb E}\left[ #1 \right]}}

\newcommand{\cc}[1]{{\ensuremath{\overline{#1}}}} % complex conjugation

\DeclareMathOperator{\supp}{supp}

\DeclareMathOperator{\erfc}{erfc}

\newcommand{\PSI}{\ensuremath{ {\mathcal S}}}
\newcommand{\PSIo}{\ensuremath{ {\mathcal S}_0}}

\newcommand{\namen}[1]{{\textsc{#1}}}           % englische w{\"o}rter als kursiv schreiben
\newcommand{\nesp}{\ensuremath{ \eta}}           % englische w{\"o}rter als kursiv schreiben
\newcommand{\tnesp}{\ensuremath{ \tilde{\eta}}}           % englische w{\"o}rter als kursiv schreiben
\newcommand{\SFCC}{\ensuremath{S_{\!F}}}           % englische w{\"o}rter als kursiv schreiben
\newcommand{\Sfcc}{\SFCC}           % englische w{\"o}rter als kursiv schreiben
\ifx \@paragraph \@empty
\makeatletter
\renewcommand\paragraph{\@startsection
{paragraph}{4}{\z@}{-3.5ex plus-1ex minus-.2ex}%
{1.3ex plus.2ex}{\normalfont\itshape}}
\makeatother
\fi

\DeclareUnicodeCharacter{04B4}{\CYRTETSE}
\renewcommand{\thmref}[1]{Theorem~\ref{#1}}
\renewcommand{\secref}[1]{Section~\ref{#1}}
\renewcommand{\lemref}[1]{Lemma~\ref{#1}}

\renewcommand{\figref}[1]{Fig.~\ref{#1}}

\newtheorem{thm}{Theorem}

\newtheorem{prop}{Proposition}
\newtheorem{defi}{Definition}

\newenvironment{mycomment}%
{%\endgraf%
\comment\noi}%
{\endcomment}
\definecolor{gray}{rgb}{0.3,0.3,0.3}
{\color{black}}                      % change to gray if you want see any changes
{\color{black}}

\newcommand{\changes}{\color{black}} % change to blue if you want see any changes
\newcommand{\changee}{\color{black}} 
\begin{document}
%%%%%%%%%%%%%%%%%%%%%%%%%%%%%%%%%%%%%%%%%%%%%%%%%%%%%%%%%%%%%%%%%%%%%%%%%%%%%%%%%%%%%
%%%%%%%%%%%%%%%%%%%%%%%%%%%%%%%%%%%%%%%%%%%%%%%%%%%%%%%%%%%%%%%%%%%%%%%%%%%%%%%%%%%%%
\setcounter{MaxMatrixCols}{12}                              % damit matrizen mit bis zu 12x12 eingebar sind

%\onecolumn
%\tableofcontents

\begin{frontmatter}
\title{Approximation of Löwdin Orthogonalization to a Spectrally Efficient Orthogonal Overlapping PPM Design for UWB
Impulse Radio }

\author[rvt]{Philipp Walk}
\ead{philipp.walk@tum.de}
\author[rvt2]{Peter Jung}
\ead{peter.jung@mk.tu-berlin.de}
\fntext[]{Some of the results of this paper were previously announced, without proofs, at the conferences
\cite{WJ10,WJT10}.}

\address[rvt]{TU-München, Chair for Theoretical Informationtechnology, Arcsistrasse 21, 80290 München}
\address[rvt2]{TU-Berlin, Chair for Information Theory and Theoretical Informationtechnology, Einsteinufer 25, 10587 Berlin}
%\author{Revision: \svnrev}

\begin{abstract}
In this paper we consider the design of spectrally efficient time-limited pulses for ultra-wideband  (UWB) systems using
an overlapping pulse position modulation scheme.  
For this we investigate an orthogonalization method, which was developed in 1950 by Löwdin \cite{Low50,Loe70}.
Our objective is to obtain a set of $N$ orthogonal (Löwdin) pulses,  which remain time-limited and spectrally efficient
for UWB systems, from a set of $N$ equidistant translates of a time-limited optimal spectral designed UWB pulse.
We derive an approximate Löwdin orthogonalization (ALO) by using circulant approximations for the Gram matrix to obtain
a practical filter implementation \changes as a tapped-delay-line \cite{TDLWG06}. \changee  We show
that the centered ALO and Löwdin pulses converge pointwise to the same \changes square-root Nyquist pulse \changee as
$N$ tends to infinity. The set of translates of the \changes square-root Nyquist pulse \changee forms an orthonormal
basis for the shift-invariant-space generated by the initial spectral optimal pulse. The ALO transformation provides a
closed-form approximation of the Löwdin transformation, which can be implemented in an analog fashion without the need of
analog to digital conversions.  Furthermore, we investigate the interplay between the optimization and the
orthogonalization procedure by using methods from the theory of shift-invariant-spaces. Finally we relate our results to wavelet and frame theory. 
\end{abstract} 
%

%\maketitle

\end{frontmatter}
%%%%%%%%%%%%%%%%%%%%%%%%%%%%%%%%%%%%%%%%%%%%%%%%%%%%%%%%%%%%%%%%%%%%%%%%%%%%%%%%%%%%%
%%%%%%%%%%%%%%%%%%%%%%%%%%%%%%%%%%%%%%%%%%%%%%%%%%%%%%%%%%%%%%%%%%%%%%%%%%%%%%%%%%%%%

%%%%%%%%%%%%%%%%%%%%%%%%%%%%%%%%%%%%%%%%%%%%%%%%%%%%%%%%%%%%%%%%%%%%%%%%%%%%%%%%%%%%%%%%%%%%%%%%%%%%%%%%%
\section{Introduction}
%%%%%%%%%%%%%%%%%%%%%%%%%%%%%%%%%%%%%%%%%%%%%%%%%%%%%%%%%%%%%%%%%%%%%%%%%%%%%%%%%%%%%%%%%%%%%%%%%%%%%%%%

% UWB setting: generall approaches

We consider in this work high data rate transmission in the ultra wideband (UWB) regime.  To prevent disturbance of
existing systems, e.g. GPS and UMTS, the Federal Communications Commission (FCC) released \cite{FCC02} a very low power
spectral density (PSD) mask for ultra-wideband (UWB) systems.  To ensure that sufficiently high \changes signal-to-noise ratio
(SNR) \changee is maintained in the frequency band $F=[0,14]$GHz, as required by the FCC, the pulses have to be designed for a
high efficient frequency utilisation. This utilisation can be expressed by the pulses normalized effective signal power
(NESP) \cite{LYG03}.  Several pulse shaping methods for pulse amplitude and pulse position modulation (PAM and PPM) were
developed in the last decade based on a FIR prefiltering of a fixed basic pulse. \changes Due to the high sampling rates
in UWB this FIR filtering is realized by a tapped delay line \cite{WTDG06}. \changee A SNR optimization under the FCC
mask constraints then reduces to a FIR filter optimization \cite{TDLWG06,WBV99,LYG03,WTDG06}. 

Since the SNR is limited, the amount of signals can be increased to achieve higher data rates or to enable multi-user
capabilities. For coherent and synchronized transmission over memoryless AWGN channels, an increased number $N$ of
mutually orthogonal UWB signals inside the same time slot, known as  $N$-ary orthogonal signal design, improves the BER
performance over $E_b/N_0$ and hence the achievable \label{rev1rem1}rate of the system
  \cite[Ch.4]{SHL94}. 

Combining spectral shaping and orthogonalization is an inherently difficult problem being neither linear nor convex.
Therefore most methods approach this problem {\bfseries sequentially}, e.g. combining spectral optimization with a
\namen{Gram-Schmidt} construction \cite{And05,WTDG06,Pro01}.  This usually results in an unacceptable loss in the NESP
value of the pulses \cite{WTDG04b,DK07}.  Moreover these orthogonal pulses are different in shape and therefore not
useful for PPM.
Furthermore, a big challenge in UWB impulse radio (UWB-IR) implementation are high rate sampling operations. Therefore,
an analog transmission scheme is desirable \cite{PCWD03,RM07,DK07} to avoid high sampling rates in AD/DA conversion
\cite{PCWD03}.

Usually, PPM is referred to an orthogonal (non-overlapping) pulse modulation scheme. To achieve higher data rates in
PPM, pulse overlapping was already investigated in optical communication \cite{BK84} and called OPPM. An application to
UWB was studied for the binary case with a Gaussian monocycle \cite{ZMSF01}. To the authors' knowledge, no orthogonal
overlapping PPM (OOPPM) signaling has been considered based on strictly time-limited pulses.  In this work we propose a
new analog pulse shape design for UWB-IR to enable  an almost  OOPPM signaling which approximate the
OOPPM scheme up to a desired accuracy.

In our approach, we first design a time-limited spectral optimized pulse $p$ and perform afterwards a Löwdin
orthogonalization of the set of $2M\!+\!1$ integer translates $\{p(\cdot\!-\!k)\}_{k=-M}^M$\changes, which span the
function space
$\SI^M(p)$\changee.  This orthogonalization
method provides an implementable and stable approximation $\poM$ to a \changes normalized \emph{square-root Nyquist
pulse} $\po$ (all integer translates are mutually orthonormal) \changee .  Using the \emph{Fourier transformation} for
$p$ given by $\hp(\nu)=\int_\R p(t)e^{-i 2\pi t \nu}dt$\label{eq:fouriertransform}
for every $\nu\in\R$, the \changes square-root Nyquist pulse \changee $\po$ can be expressed in the frequency-domain
\begin{align}
  \hpo(\nu) = \frac{\hp(\nu)}{\sqrt{\sum_k \Betrag{\hp(\nu -k )}^2}}\label{eq:nyquist} 
\end{align}
for $\nu$ almost everywhere, which is known as the \emph{orthogonalization trick} \cite{Dau90}. 

Usually in digital signal processing, see \figref{fig:AppScheme}, the time-continuous signals (pulses) \changes in $L^2$
\changee are sampled by an analog to digital (AD) operation \changes (either in time or frequency) \changee to obtain
time-discrete signals in $\C^{2M+1}$. Then a discretization of the orthogonalization in \eqref{eq:nyquist} yields a
finite digital transformation $D^M$ to construct a discrete signal in the \changes frequency \changee domain which has
again to be transformed by a DA operation to obtain finally an approximation of the \changes square-root \changee
Nyquist pulse \cite{JS02}.\\
Instead of using such an AD/DA conversion to operate in a discrete domain \changes (depictured in \figref{fig:AppScheme} by the
digital box)\changee, we use the {\bfseries democratic}
\label{rev1remii} Löwdin 
orthogonalization\changes\footnote{In fact, the Löwdin orthogonalization and the orthogonalization trick are both
orthonormalizations, but for historical reasons we will refer them as orthogonalization methods.} \changee $B^M$
\changes (in the analog box)\changee, found by Per-Olov Löwdin in \cite{Low50,Loe70}, where all $2M+1$ linear
independent pulse translates are involved simultaneously by a linear combination to generate a set of time-limited
mutual orthonormal pulses. Hence the Löwdin orthogonalization is order independent. The Löwdin pulses constitute then an
orthonormal basis for the span $\SI^M(p)$ of the initial basis $\{p(\cdot -k)\}_{k=-M}^M$.  Moreover, as we will show in
our main Theorem~\ref{th:lozak},  the Löwdin orthogonalization  $B^M$ is a stable approximation method to \changes the
construction \changee in \eqref{eq:nyquist} and operates completely in the analog domain. Note, that a discretization of
\eqref{eq:nyquist} generates neither shift-\changes orthonormal \changee pulses nor a set of mutual \changes orthonormal
\changee pulses.  The important property of the Löwdin method is its minimal summed energy distortion to the initial
basis. It turns out that all \changes orthonormal \changee pulses maintain  the spectral efficiency ''quite well''.

As $M$ tends to infinity the Löwdin orthogonalization $B=B^\infty$ applied to the initial pulse $p$ delivers the
\changes square-root \changee Nyquist pulse $\po$, which allows a real-time OOPPM %
system with only one single matched filter at the receiver. Since the Löwdin transform $B^M$ is hard to compute and to
control we introduce an approximate Löwdin orthogonalization (ALO) $\tilde{B}^M$ and investigate its stability and
convergence properties. It turns out that for fixed $M$ the transformations $B^M$ and $\tilde{B}^M$ are both
implementable by a FIR filter bank \changes (realised as a tapped delay line) like the  spectral  optimization in
\cite{WTDG06}\changee. We call $B^M p$ and $\tilde{B}^M p\ $
\emph{approximate \changes square-root \changee Nyquist} pulses, since we observed that even for finite $M$ our analog
approximation yields time--limited pulses with almost shift--orthogonal character since the sample values of the
autocorrelation are below a measurable magnitude for the correlator. Hence such a construction of approximate \changes
square-root \changee Nyquist pulses seems to be promising  for an OOPPM system. 
\begin{figure}
  \begin{center}
  \includegraphics[scale=0.35]{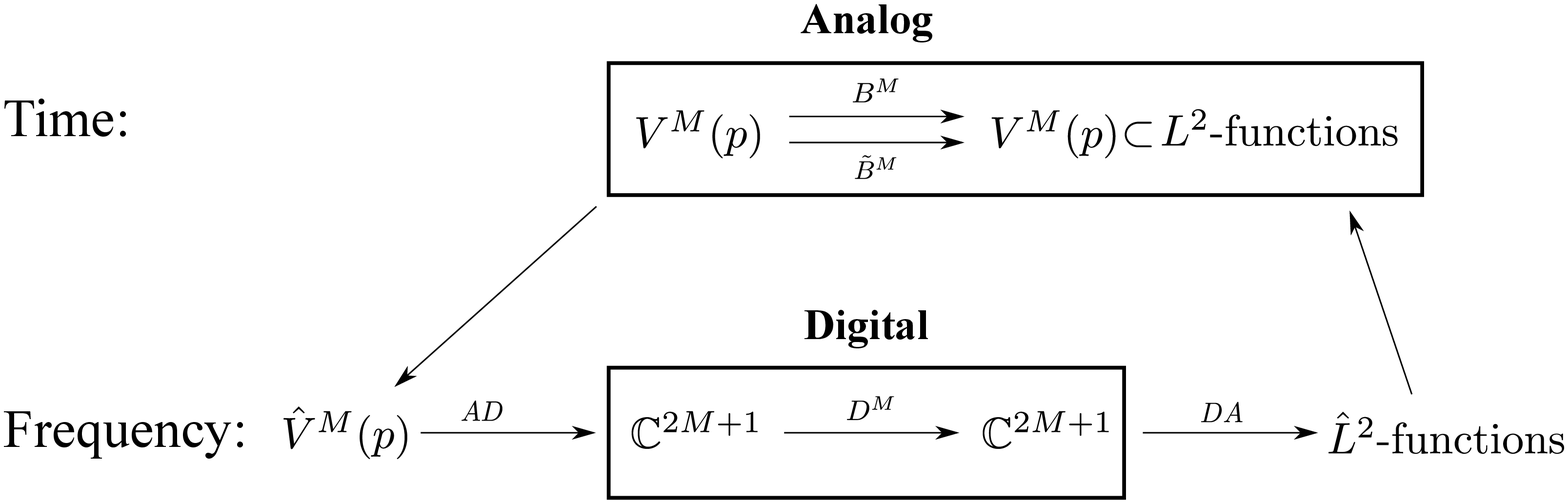}
  \caption{Analog and discrete approximation methods in time and frequency domains}\label{fig:AppScheme}
  \end{center}
\end{figure}

The structure of this paper is as follows: In \secref{sec:signalmodel} we introduce the signal model and motivate our
spectrally efficient $N$-ary orthogonal overlapping PPM design for UWB systems. %
\secref{sec:pulseshapedesign} presents the state of the art in FCC optimal pulse shaping for UWB-IR based on PPM or PAM
transmission by FIR prefiltering of Gaussian monocycles which is a necessary prerequisite for our design.  %
To develop our approximation and convergence results in  \secref{sec:orthogonalization} we introduce the theory of
shift--invariant spaces in \secref{sec:shiftspaces} and in \secref{sec:lofinite} the Löwdin orthogonalization for a set
of $N$ translates to provide a $N$-ary orthogonal overlapping transmission.  Our main result is given in
\secref{sec:loinfinite}, where we consider the stability  of the Löwdin orthogonalization $B^M$ (for $M$ increasing to
infinity) and develop for this a simplified approximation method $\tilde{B}^M$, called ALO. %  
In \secref{sec:discussion} we study certain properties of our filter design and investigate the combination of both
approaches.  Furthermore, in \secref{sec:TFandONB} we develop a connection between our result and the canonical tight
frame construction.  Finally in \secref{sec:app}, we demonstrate that the ALO and Löwdin transforms yields for
sufficiently large filter orders compactly supported approximate \changes square-root \changee Nyquist pulses, which can
be used for OOPPM having high spectral efficiency in the FCC region. Moreover,  the Löwdin pulses provides also a
spectrally efficient $(2M+1)$-ary orthogonal pulse shape modulation (PSM) \cite{GK08}.

%%%%%%%%%%%%%%%%%%%%%%%%%%%%%%%%%%%%%%%%%%%%%%%%%%%%%%%%%%%%%%%%%%%%%%%%%%%%%%%%%%%%%%%%%%%%%%%%%%%%%%%%%
\section{Signal Model}\label{sec:signalmodel}
%%%%%%%%%%%%%%%%%%%%%%%%%%%%%%%%%%%%%%%%%%%%%%%%%%%%%%%%%%%%%%%%%%%%%%%%%%%%%%%%%%%%%%%%%%%%%%%%%%%%%%%%%

In this work we will consider finite energy pulses $p$, i.e. the set $L^2:=\set{p:\R\to\C}{\Norm{p}_{L^2} <\infty}$ of
 square-integrable  functions with norm $\Norm{\cdot}_{L^2}:=\sqrt{\SP{\cdot}{\cdot}}$ induced by the
inner product\footnote{The bar denotes the closure for sets and the complex conjugation for vectors or functions.} 
\begin{align}
  \SP{p}{q}:=\int_\R p(t) \cc{q(t}) dt \label{eq:sp}.
\end{align}

To control signal power in time or frequency locally we need \changes pulses in
$L^\infty:=\set{p:\R\to\C}{\Norm{p}_{\infty} <\infty}$ which are  essentially  bounded, i.e. functions with a
finite $L^\infty$-norm, given as \changee
\begin{align}
  \Norm{p}_\infty := \esssup{t\in\R}{ | p(t)|},\label{eq:esssup}
\end{align}
where the \emph{essential supremum} is defined as the smallest upper bound for $\Betrag{p(t)}$ except on a set of
measure zero.  If the pulse is continuous than this implies boundedness everywhere.
UWB-IR technology uses ultra short pulses, i.e. strictly time-limited pulses with support contained in a finite interval
$X\subset\R$. We call such  $L^2$- functions \emph{compactly supported} in $X$ and denote its 
closed  span by the subspace $L^2(X)$.  The coding of an information sequence $\{d_n\}=\{d_n\}_{n\in\Z}$ is
realized by pulse modulation techniques \cite{Pro01,Mid96} of a fixed normalized basic pulse $p\in L^2([0,T_p])$ with
duration $T_p$.

A relevant issue in the UWB-IR framework and in our paper is the spectral shape of the pulse.  In this section we will
therefore summarize the derivation of spectral densities for common UWB modulation schemes such as PAM \cite{TDLWG06},
PPM \cite{Sch93,LYG03,WS98a,Win99,RS98} and combinations of both \cite{NM03} to justify our spectral shaping in the next
section. % 
Antipodal PAM and $N$-ary PPM are linear modulation schemes which map each data symbol $d_n$  to a pulse (symbol)
$s_{d_n}(t)$ with the same power spectrum $\En\Betrag{\hp(\nu)}^2$.  If we fix the energy $\En$ of the transmitted
symbols and the pulse repetition time (symbol duration) $T_s$, we will show now for certain discrete random processes
( e.g. i.i.d.  processes \cite{WJ10}) that the power spectrum density (PSD) of the transmitted signals
is given by
\begin{align}
  S_{u} (\nu) = \frac{\En \Betrag{\hp(\nu)}^2}{T_s}\label{eq:psd}.
\end{align}
Hence an optimization of the pulse power spectrum  to the FCC mask $\SFCC(\nu)$ over the
band $F$ in \secref{sec:pulseshapedesign} increases the transmit power.  To be more precise, PPM
produces discrete spectral lines, induced by the periodic pulse repetition, the use of uniformly
distributed pseudo-random time hopping (TH) codes $c_n\in[0,N_c]$ was suggested to reduce this effect and to enable
multi-user capabilities \cite{Sch93,Win99,WS98a,Win02,WZK08}:
\begin{align}
  u(t) = \sum_{n=-\infty}^\infty \sqrt{\En} p (t - n T_f - c_n T_c - d_{\lfloor n/N_f \rfloor} T ) \label{eq:uwbsignal}.
\end{align}
In \cite{Win02} this is called framed TH by a random sequence, since the coding is repeated in each frame $N_f$ times
with a clock rate of $1/T_f$.  Hence $N_f T_f = T_s$ is the symbol duration for transmitting one out of $N$ symbol
waveforms representing the encoded information symbol $d_n\in\{0,\dots,N-1\}$.  To prevent ISI and collision with other
users, the maximal PPM shift $T$ and TH shift $T_c$ have to fulfill $NT\leq T_c$ and $N_cT_c\leq T_f$. To ensure mutual
orthogonality of all symbols one requires $T>T_p$. The PSD for independent discrete i.i.d. processes $\{c_n\},\{d_n\}$
follows from the Wiener-Khintchine Theorem \cite{Win02,VE05} to \cite[(5)]{LYG03}, \cite{NM03,NM06}.
\begin{align}
  S_{u} (\nu) = \frac{\En\Betrag{\hp(\nu)}^2}{T_f} \left[ 1\!-\!\Betrag{G_{\beta}(\nu)}^2 
  + \frac{\Betrag{G_{\beta}(\nu)}^2}{T_f}\sum_k \delta\left(\nu \!-\!\frac{k}{T_f}\right) \right] 
\end{align}
\begin{align}
  \text{with } 
  \Betrag{G_\beta (\nu)} = 
  \frac{1}{N_c N} \Betrag{\frac{\sin(\pi \nu T_c N_c)}{\sin(\pi \nu  T_c)}} \Betrag{\frac{\sin(\pi \nu
  T N)}{\sin(\pi \nu  T)}} \label{eq:gs} \text{ and $\delta$ is the Dirac distribution}.
\end{align}
However, a more effective and simple reduction method without the use of frame repetition ($N_f=1$) or random TH
($N_c=1$) has been proposed in \cite{NM03,NM06}.  Here antipodal PAM with $a_n\in\{-1,1\}$ is combined with $N$-ary PPM
modulation, for $NT\leq T_s$ 
\begin{align}
  u(t)= \sqrt{\En}\sum_n a_n p(t-nT_s - d_nT)\label{eq:APAMPPM}.
\end{align}
The PSD for such i.i.d. processes is well known \cite[Sec.4.3]{Mid96}:
\begin{align}
  S_{u}(\nu) = \frac{\En \Betrag{\hp(\nu)}^2}{T_s} \Bigg[ \Expect{a^2} - \Betrag{\Expect{a} \cdot \Expect{e^{-i 2\pi
  \nu d T  }}}^2 +\frac{ \Betrag{\Expect{a} \cdot \Expect{e^{-i 2\pi \nu dT }}}^2}{T_s} \sum_n
  \delta\left(\nu -\frac{n}{T_s} \right) \Bigg]\label{eq:bppam}\ ,
\end{align}
since we have $G_\beta(\nu)=\Expect{a} \Expect{e^{i2\pi \nu dT }}$ in \eqref{eq:gs}. For an i.i.d. process $\{a_n\}$
with expectation $\Expect{a}=0$ and variance $\Expect{a^2}=1$ the PSD reduces to \eqref{eq:psd}.  Hence the effective
radiation power is \emph{essentially determined by the pulse shape} times the energy $\En$ per symbol duration $T_s$ and
should be bounded pointwise on $F$ below the FCC mask $\SFCC$
\begin{align} 
  S_u (\nu) = \En \frac{\Betrag{\hp (\nu)}^2}{T_s} \leq \SFCC(\nu)
  \label{eq:psdeirp}\quad \text{for all} \quad\nu\in F.
\end{align}
The optimal receiver for $N$-ary orthogonal PPM in a memoryless AWGN  channel with noise power density $\Noise$ is the
coherent correlation receiver. The  \changes uncoded \changee bit \label{rev1mem1b} rate $R_b$ and average symbol error
  probability $P_s$ is given as \cite[4.1.4]{SHL94}
\begin{align}
  R_b =\frac{\log {N}}{T_s}\quad \text{and}\quad P_s (\En) \leq (N-1) \erfc\left(
  \sqrt{\frac{\En}{\Noise}}\right),\label{eq:performpar}
\end{align}
where $\erfc$ is the complementary error function. Hence, a performance gain for fixed $T_s$ is achieved by increasing
$N$ and/or $\En$.
\paragraph[:]{Increasing  $N$} Usually non--overlapping pulses are necessary in PPM to guarantee orthogonality of the
set  $\{p(\cdot-nT)\}:=\{p(\cdot - nT)\}_{n\in\Z}$ of pulse translates, i.e. $T>T_p$. For fixed $T_p$ this limits the
number of pulses $N$ in $[0,T_s]$ and hence the rate $R_b$. In this work we will design an orthogonal overlapping PPM
(OOPPM) system by keeping all overlapping translates mutually orthogonal.  But such \changes square-root \changee
Nyquist pulses are in general not time--limited, i.e. not compactly supported. In fact, we will show that for a
particular class of compactly supported pulses a non-overlapping of the \changes translates \changee is necessary to
obtain strict shift-orthogonality.  However, we derive overlapping compactly supported pulses approximating the \changes
square-root \changee Nyquist pulse in \eqref{eq:nyquist} and characterize the convergence.  These \changes
approximations to the square-root \changee Nyquist pulse allow a realizable $N$-ary OPPM implementation based on FIR filtering of
time--limited analog pulses.

\paragraph{Increasing $\mathcal{E}$} The maximization of $\En$ with respect to the FCC mask was already studied in
\cite{LYG03,TDLWG06} where a FIR prefiltering is used to shape the pulse  such that its radiated power spectrum
efficiently  exploits and strictly respect the FCC mask. Note that the FCC regulation in \eqref{eq:psdeirp} is a local
constraint and does not force a strict band-limited design, however fast frequency decay outside the interval $F$ is
desirable for a hardware realization.\\  

Our combined approach now relies on the construction of two prefiltering operations to shape a fixed initial pulse. The
first filter shapes the pulse to optimally exploit the FCC mask and the second filter generates an \changes
approximation to the square-root \changee Nyquist pulse. The filter operations can be described as \emph{semi--discrete
convolutions} of pulses $p\!\in\!  L^2$  with sequences \changes $\vc\in\ell^2:=\{\va \mid \Norm{\va}^2_{\ell^2} \!=\!
\sum_n \Betrag{a_n}^2 \!<\!\infty\}$ \changee
\begin{align} 
  p \convdisT \vc := \sum_{n\in\Z} c_n  p(\cdot - nT) \label{eq:semiconv},
\end{align}
\changes at \changee clock rate $1/T$. If we restrict ourselves to FIR filters of order $L$ \changes  the impulse
response becomes a sequence $\vc\in\ell^2(L)$  which can be regarded as a vectors in $\C^L$. Hence we refer for $\vc\in
\C^L$ to FIR filters and for $\vc\in \ell^2$ to IIR filters.\changee

%%%%%%%%%%%%%%%%%%%%%%%%%%%%%%%%%%%%%%%%%%%%%%%%%%%%%%%%%%%%%%%%%%%%%%%%%%%%%%%%%%%%%%%%%%%%%%%%%%%%%%%%%
\section{FCC Optimization of a Single Pulse}\label{sec:pulseshapedesign}
%%%%%%%%%%%%%%%%%%%%%%%%%%%%%%%%%%%%%%%%%%%%%%%%%%%%%%%%%%%%%%%%%%%%%%%%%%%%%%%%%%%%%%%%%%%%%%%%%%%%%%%%%

The first prefilter operation generates an optimized FCC pulse $p$. To generate a time-limited real-valued pulse we
consider a real-valued initial input pulse $q\in L^2([-T_q/2,T_q/2])$ and a real-valued (causal) FIR filter
$\vgL\in\R^L$. A common UWB pulse is the truncated Gaussian monocycle $q$ \cite{LYG03,TDLWG06,BEZJ07}, see also
\secref{sec:app}. The prefilter
operation is then:
\begin{align}
  p(t) =(q \convdisTo \vg) (t) = \sum_{k=0}^{L-1} g_k q(t-kT_0) \label{eq:firfilter}
\end{align}
which results in a maximal duration \changes (support length) \changee $T_p=(L-1)T_0 + T_q$ of $p$.  

To maximize the PSD according to \eqref{eq:psdeirp}  we have to shape the initial pulse by the filter $\vg$ to exploit
efficiently the FCC mask $\SFCC$ in the passband $F_p\!\subset\! F$, i.e. to maximize  the ratio of the
pulse power in $F_p$ and the maximal power allowed by the FCC 
\begin{align}
  \nesp (p) := \int_{F_p} |\hp(\nu)|^2 d\nu \bigg/ \int_{F_p} \Sfcc(\nu) d\nu. \label{eq:nesp}
\end{align}
This is known as the direct maximization of the NESP value $\nesp(p)$, see \cite{WTDG06}.  Here we already
included the constants $\En$ and $T_s$ in the basic pulse $p$.  If we fix the initial pulse $q$, the clock rate $1/T_0$
and the filter order $L$, we get the following  optimization problem 
\begin{align}
  \begin{split}
    &\max_{\vg \in \R^L} \tnesp \left( q \convdisTo \vg \right)\\
    \text{s.t. } &\forall \nu\in F : \Betrag{\hvg(\nu)}^2\cdot \Betrag{\hq(\nu)}^2 \leq \Sfcc(\nu),
  \end{split}\label{eq:problemsub}
\end{align}
where $\hvg$ denotes the $1/T_0$ periodic  Fourier series of $\vg$, which is defined for an arbitrary sequence
$\vc\in\ell^2$ as
\begin{align}
  \hat{\vc}(\nu)= \sum_{n=-\infty}^{\infty} c_n e^{-2\pi i \nu n T_0}\label{eq:dft}.
\end{align}
Since $\vg\in\ell^2(L)$ the sum in \eqref{eq:dft} becomes finite for $\hvg$. % 
Liu and Wan \cite{LW08} studied the non-convex optimization problem \eqref{eq:problemsub} with non-linear constraints
numerically  with fmincon, a \namen{Matlab} program. The disadvantage of this approach lies in the trap of a local
optimum, which can only be overcome by an intelligent choice of the start parameters.

Alternatively \eqref{eq:problemsub} can be reformulated in a convex form by using the Fourier series of the
autocorrelation $r_{\vg,n}:=\sum_k g_k g_{k-n}$ of the filter $\vg$ \cite{DLS02}.   Since $r_{\vg,n}=r_{\vg,-n}$
(real-valued symmetric sequence) we write
\begin{align}
  \hvr(\nu) &:=  \sum_{n=0}^{L-1} r_{\vg,n} \phi_n(\nu) = \Betrag{\hvg(\nu)}^2
\end{align}
on the frequency band $[-\frac{1}{2T_0},\frac{1}{2T_0}]$ by using the basis $\vphi :=\left\{1, 2\cos(2\pi \nu T_0),2\cos
(2\pi \nu 2 T_0),\dots\right\}$ and get $\Abs{\hp(\nu)}^2 = \hvr (\nu)\cdot \Abs{\hq(\nu)}^2$. Due to the symmetry of
$\phi_n$ and $\SFCC$ we can restrict the constraints in \eqref{eq:problemsub} to $F=[0,\frac{1}{2T_0}]$ and 
obtain the following semi--infinite linear problem:
\begin{align}
  &\max_{\vr \!\in \R^{L}} \sum_{n=0}^{L-1} r_{\vg,n} c_n\quad\!\text{such that} \quad  0 \leq
  \hvr(\nu) \leq \M(\nu)\ \text{ for all $\ \nu\!\in\! \left[0,\frac{1}{2T_0}\right]$}\\
  &\text{with}\quad\M(\nu):=\frac{\SFCC(\nu)}{\Abs{\hq(\nu)}^2}\quad \text{and}\quad c_n :=
  \int_{F_p} \Betrag{\hq(\nu)}^2 \phi_n (\nu) d\nu.\label{eq:constrains}
\end{align}
Since the FCC mask  is piecewise constant, we separate $\M(\nu)$ into five sections $\M_i(\nu)$ \cite{BEJJ06} and get
the inequalities%
\begin{align}
  \forall i={1,\dots,5}:\ \hvr (\nu) &\leq \M_i (\nu)\quad\text{for}\quad \nu\in [\alpha_i,\beta_i]  \label{eq:semicons}
\end{align}
with $\beta_1=1.61, \beta_2=1.99, \beta_3=3.1, \beta_4=10.6, \beta_5=14$ and $\alpha_1=\alpha_2=\alpha_3=\alpha_4=0,
\alpha_5=\beta_4$
in GHz, see \figref{fig:fit}. \\
The necessary lower bound for $\vr$ reads 
\begin{align}
  \hvr (\nu) \geq 0 \quad\text{for}\quad \nu\in \left[0,\frac{1}{2T_0}\right]=[0,14] \text{GHz}\label{lbound}.
\end{align}
To formulate the  constraints in \eqref{eq:constrains} for $\vr$ as a positive bounded cone in $\R^L$ we approximate
 $\M_i(\nu)$
by trigonometric polynomials\footnote{
Since the FCC mask divided by the Gaussian power spectrum is monotone increasing from $0$ to $10.6$GHz we can let
$\Gamma_1, \dots, \Gamma_4$  overlap.}
$\Gamma_i(\nu):=\sum_n \gamma_{i,n} \phi_n (\nu)$ of order $L$ in the $L^2$-norm \cite{BEJJ06}.
The semi-infinite linear constraints in \eqref{eq:semicons} describe a compact convex set \cite[(40),(41)]{DLS02}. To
see  this, let us introduce the following lower bound cones for $\theta \in [0,\frac{1}{2T_0}]$
\begin{align}
  K_{\text{low}}(\theta)&= \Set{\vr  \in \R^L \bigg| \sum_{k=0}^{L-1} r_{\vg,n} \phi_n(\nu) \geq 0 ,  \nu\in
  \left[\theta,\frac{1}{2T_0}\right]}.
\end{align}
For $\theta=0$ the positive cone $ K_0=K_{\text{low}}(0)$ defines the lower bound in \eqref{lbound} if we set
$T_0=\frac{1}{28\text{GHz}}$.  To formulate the non-constant upper bounds, one can use the approximation functions
$\Gamma_i(\nu)$ \cite{BEJJ06} given in the same basis $\vphi$ as  $\Betrag{\hvg(\nu)}^2$.  For each
$i\in\{1,\dots,5\}$ the bounds in \eqref{eq:semicons} are then equivalent to 
\begin{align}
  \sum_{n=1}^L (\gamma_{i,n} - r_{\vg,n} )\phi_n(\nu) \geq 0\quad\text{for}\quad \nu\in [\alpha_i,\beta_i].
\end{align}
For the upper bounds, we just have to set $\rho_{i,n}:=\gam_{i,n} - r_{\vg,n}$ for each $i=1,\dots,5$ and $n\geq 1$,
which leads to the upper bound cones 
\begin{align}
  K_{\text{up}}(\theta_i)&= \bigg\{\vr \in \R^L \bigg|  \sum_{n=0}^{L-1} \rho_{i,n} \phi_n (\nu)\geq 0, \nu\in
  \left[\theta_i,\frac{1}{2T_0}\right] \bigg\}, \\
  \bar{K}_{\text{up}}(\theta_i)&= \bigg\{\vr  \in \R^L\bigg|  \sum_{n=0}^{L-1} \rho_{i,n} \phi_n (\nu) \geq 0, \nu
  \in \left[0,\theta_i  \right]  \bigg\}.
\end{align}
The five upper bound cones $K_i$ are then 
\begin{align}
  \forall i=1,\dots,4: \ K_i &:= \bar{K}_{\text{up}} (\beta_i) \quad \text{and} \quad K_5 := {K}_{\text{up}} (\alpha_5).
\end{align}
Since the autocorrelation has to fulfill all these constraints, it has to be an element of the intersection.  After this
approximation\footnote{
The $\Gamma_i$ are approximations to the FCC mask with a certain error.  Also, $T_0$ is now fixed via the
frequency range $F$. If one wants to reduce $T_0$, one has to reformulate the cones, hence $\gamma_{i}$ and extend
the frequency band constraints. Increasing $T_0$ above $1/28\text{GHz}$ is not possible, if one wants to
respect the whole mask.}
we get from \eqref{eq:problemsub} the problem
\begin{align}
  \max_{\vr\in \bigcap_{i} K_i} \sum_{n=0}^{L-1} r_{\vg,n} c_n.\label{problemSDP}
\end{align}
This is now a convex optimization problem of a linear functional over a convex set.  By the \emph{positive real lemma}
\cite{DLS02}, these cone constraints can be equivalently described by semi-positive-definite matrix equalities, s.t. the
problem \eqref{problemSDP} is numerically solvable with the \namen{matlab} toolbox SeDuMi \cite{WBV99,Stu99}.  The
filter is obtained by a spectral factorization of $r_{\vg}$. Obviously $\vg$ is not uniquely determined.
\begin{figure}[ht]
  \begin{center}
  \includegraphics[scale=0.47,angle=-90]{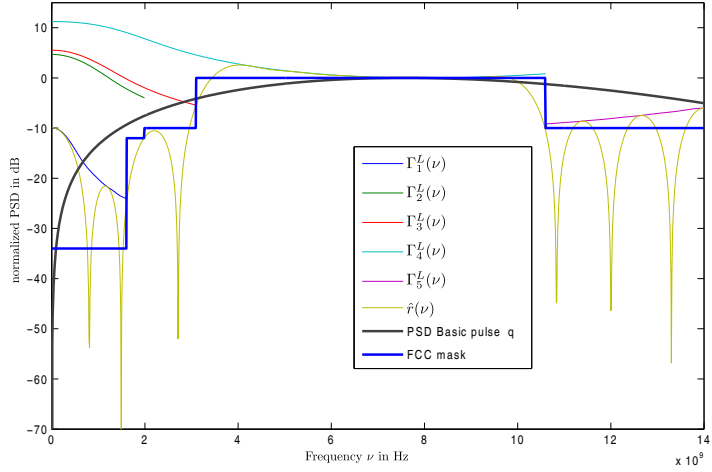}
  \caption{Fourier-approximations $\Gamma_i^L$ of $\mathcal{M}$ for $L=25$.}\label{fig:fit}
  \end{center}
\end{figure}

\noi Note that this optimization problem can also be seen as the maximization of a local $L^2$-norm, given as the NESP
value, under the constraints of local $L^\infty$-norms. %

%%%%%%%%%%%%%%%%%%%%%%%%%%%%%%%%%%%%%%%%%%%%%%%%%%%%%%%%%%%%%%%%%%%%%%%%%%%%%%%%%%%%%%%%%%%%%%%%%%%%%%%%%
\section{Orthogonalization of Pulse Translates}\label{sec:orthogonalization}
%%%%%%%%%%%%%%%%%%%%%%%%%%%%%%%%%%%%%%%%%%%%%%%%%%%%%%%%%%%%%%%%%%%%%%%%%%%%%%%%%%%%%%%%%%%%%%%%%%%%%%%%%

In \cite{TDLWG06,DK07} a sequential pulse optimization was introduced, which produces mutually orthogonal pulses $\po_m=
q \convdisTo \vgm$, i.e. $\SP{\po_m}{\po_n}=\del_{nm}$.
Here each pulse $\po_m$ is generated by a different FIR filter $\vgm\in\R^L$, which depends on the previously generated
pulses $\po_{1},\dots, \po_{m-1}$ and produces pulses in $L^2([-T_q/2,(L-1)T_0+T_q/2])$. This approach is similar to the
Gram-Schmidt construction in that it is order-dependent, since the first pulse $\po_1$ can be optimally designed to the
FCC mask without an orthogonalization constraint. 
We will now present a new order--independent method to generate from a fixed initial pulse $p$ a set of orthogonal
pulses $\{\po_m\}$.  Therefore we introduce a new time--shift $T>0$, namely  the PPM shift in \eqref{eq:uwbsignal}, to
generate a set of $N=2M+1$ translates $\{ p(\cdot -mT)\}_{m=-M}^M$, i.e. $M$ shifts in each time direction. The
orthogonal pulses are then obtained by linear combinations of the translates of the initial pulse $p$.  For a stable
embedding of the finite construction we restrict the initial pulses to the set $L^2_{T_p}:=L^2([-T_p/2,T_p/2])$ of
centered pulses with finite duration $T_p$. To study the convergence we need to introduce the concept of regular
shift--invariant spaces.
%
%%%%%%%%%%%%%%%%%%%%%%%%%%%%%%%%%%%%%%%%%%%%%%%%%%%%%%%%%%%%%%%%%%%%%%%%%%%%%%%%%%%%%%%%%%%%%%%%%%%%%%%%%
\subsection{Shift--Invariant Spaces and Riesz-Bases}\label{sec:shiftspaces}
%%%%%%%%%%%%%%%%%%%%%%%%%%%%%%%%%%%%%%%%%%%%%%%%%%%%%%%%%%%%%%%%%%%%%%%%%%%%%%%%%%%%%%%%%%%%%%%%%%%%%%%%%
%
To simplify notation we scale the time axis so that $T=1$.  Let us now consider the set $\PSIo(p):=\spann\{p(\cdot
-n)\}$ of all finite linear combinations of $\{p(\cdot -n)\}$, which is certainly a subset of $L^2$. The $L^2$-closure
of $\PSIo(p)$ is a \emph{shift--invariant} closed subspace $\PSI(p):=\cc{\PSIo(p)}\subset L^2$, i.e.  for each
$f\in\PSI(p)$ also $\{ f(\cdot -n)\} \subset \PSI(p)$.   Since $\PSI(p)$ is generated by a single function $p$ we call
it a \emph{principal shift--invariant} (PSI) space  and $p$ the \emph{generator} for $\PSI(p)$. In fact, $\PSI(p)$ is
the smallest PSI closed subspace of $L^2$ generated by $p$. Of course not every closed PSI space is of this form
\cite{BDR94b}. In this work we are interested in spaces which are closed under semi-discrete convolutions
\eqref{eq:semiconv}  with $\ell^2$ sequences, i.e. the space 
\begin{align}
  \SI(p):=\set{p\convdis \vc}{\vc\in \ell^2}
\end{align}
endowed with the $L^2$- norm. 
  Note that $\SI(p)$ is in general not a subspace or even a closed subspace of $L^2$ \cite{BDR94b,AS02b}. But to
  guarantee stability of our filter design $\SI(p)$ has to be closed, i.e. has to be a Hilbert subspace. More
  precisely, the translates of $p$ have to form a \emph{Riesz basis}. 
%

%
%%%%%%%%%%%%%%%% DEFINITION: 

  \begin{defi}\label{def:rieszbasis}
  Let $\Hil$ be a Hilbert space. $\{e_n\} \subset\Hil$ is a Riesz basis for $\overline{\spann\{e_n\}}$ if
  and only if there are constants $0 < A\leq B<\infty$, s.t.
  \begin{align}
    A\Norm{\vc}^2_{\ell^2} \leq \Norm{ \sum_n c_n e_n}^2_{\Hil} \leq B \Norm{\vc}^2_{\ell^2} \quad\text{for all} \quad
    \vc\in \ell^2\label{eq:rieszconditionorig}. 
  \end{align}
\end{defi}
%%%%%%%%%%%%%%%%%%%%%%%%%%%%%%%%%%%%%%%5
In this case $\overline{\spann \{e_n\} }$ becomes a Hilbert-subspace of $\Hil$. For SI spaces in $L^2=\Hil$ we get the
following result.%

%%%%%%%%%%%%%%%% PROPOSITION: 
\begin{prop}[Prop.1 in \cite{AU94}]\label{pro:au} 
  Let $p\in L^2(\R)$. 

 Then $\SI(p)$ is a closed shift--invariant subspace of $L^2$ if and only if 
  \begin{align} A
    \Norm{\vc}^2_{\ell^2} \leq \Norm{p \convdis \vc }^2_{L^2} \leq B \Norm{\vc}^2_{\ell^2} \quad\text{for all} \quad \vc\in
    \ell^2\label{eq:rieszconditionorig2} 
  \end{align} 
  holds for fixed constants $0 < A\leq B<\infty$. Moreover, $\{p(\cdot-n)\}$ is a Riesz-basis for $\SI(p)$.

\end{prop}
%%%%%%%%%%%%%%%%%%%%%%%%%%%%%
%
If the generator $p$ fulfills \eqref{eq:rieszconditionorig2}, then $\SI(p)= \PSI(p)$ by \cite{Jia97}
and we call $p$ a \emph{stable generator} and $\SI(p)$ a \emph{regular PSI space} \cite{BDR94b}.  An
\emph{orthonormal generator} \changes (\emph{square-root Nyquist pulse, shift-orthonormal pulse}) \changee $\po$ for $\SI(p)$  is a
generator with  $\SP{\po(\cdot -n)}{\po(\cdot -m)}=\del_{mn}$ for all
$n,m\in\Z$ \cite{And05,Lap09}.  
Benedetto and Li \cite{BL98} showed that the stability and orthogonality of a generator $p\in L^2$ can be described by
the absolute $[0,1]$-integrable periodic function $\Phi_p\in L^1([0,1])$  of $p$ defined for $\nu$ almost everywhere
(a.e.) as
\begin{align}
  \Phi_p (\nu):=\sum_k \Betrag{\hp (\nu +k)}^2. \label{eq:symbol}
\end{align}        
They could show the following characterization \cite{Chr03,BL98}.
%
%%%%%%%%%%%%%%%% PROPOSITION: 
\begin{thm}[Th. 7.2.3 in \cite{Chr03}]\label{th:riesz}
    A function $p\in L^2$ is a stable generator for $\SI (p)$ if and only if
    there exists $0<A\leq B<\infty$ such that
    \begin{align} 
      A \leq &\Phi_p (\nu)  \leq B \qquad \text{for $\nu$ a.e.} \label{eq:rieszcondition}
    \end{align}
      and is an orthonormal generator for $\SI(p)$ if and only if
      \begin{align}
      &\Phi_p(\nu) =1 \qquad \text{for $\nu$ a.e.. } \label{eq:orthocondition} 
    \end{align} 
\end{thm}
%%%%%%%%%%%%%%%%%%%%%%%%%%%%%
%

\begin{proof}
For a proof see Th. 7.2.3. (ii) and (iii) in \cite{Chr03}. In our special case we have $B=1$. Note, that the Riesz
sequence and orthonormal sequence are bases for their closed span, meaning that in our case $\PSI(p)=\SI(p)$.
\end{proof}

Due to this characterization in frequency there is a simple ``orthogonalization trick'' for a stable generator given in
\eqref{eq:nyquist}, which was found by Meyer, Mallat, Daubechies  and others \cite{Mey86},\cite[Prop. 7.3.9]{Chr03}.
Unfortunately, this does not provide an ``a priori'' construction in the time domain and does not lead to a support control
of the orthonormal generator in time, as necessary for UWB-IR. 
%
%These approximated pulses $B^M p$ should be smooth, still
%FCC-optimal in some $L^2$-sense (see \eqref{eq:mep:finite} below),
%having good correlation properties and converge in the limit to a shift--invariant orthonormal pulse, i.e. to a
%Nyquist pulse.  
%

Contrary to an approximation in the frequency-domain we approach an approximation in time-domain via the Löwdin
transformation.
We will show that in the limit the Löwdin transformation for shift--sequences is in fact given in frequency by the
orthogonalization trick \eqref{eq:nyquist}.  By using finite section methods we establish  an approximation method
in terms of the discrete Fourier transform (DFT) to allow an easy computation.  Furthermore, we show that the Löwdin
construction for stable generators is unique and optimal in the $L^2$-distance among all orthonormal generators and
corresponds to the canonical construction of so called tight frames (given later).

%%%%%%%%%%%%%%%%%%%%%%%%%%%%%%%%%%%%%%%%%%%%%%%%%%%%%%%%%%%%%%%%%%%%%%%%%%%%%%%%%%%%%%%%%%%%%%%%%%%%%%%%%
\subsection{Löwdin Orthogonalization for Finite Dimensions}\label{sec:lofinite}
%%%%%%%%%%%%%%%%%%%%%%%%%%%%%%%%%%%%%%%%%%%%%%%%%%%%%%%%%%%%%%%%%%%%%%%%%%%%%%%%%%%%%%%%%%%%%%%%%%%%%%%%%

Since the Löwdin transformation is originally defined for a finite set of linearly independent elements in a Hilbert space
$\Hil$, we will use the \emph{finite section method} to derive a stable approximation to the infinite case. For this we
consider for any $M\in\N$ the symmetric orthogonal projection $\vPM$ from $\ell^2$ to $\ell^2_M=\{\vc\in \ell^2|\supp
\vc\!  \subset\!\{-\!M,\dots,M\}\}$ defined for $\vc\in\ell^2$ by $\vPM \vc := \vc^M=(0, \dots, 0,
c_{-M},\dots,c_{M},0,\dots,0)$. Then the finite section $\GramM$ of the infinite dimensional Gram matrix 
$\Gram$ of $p\in L^2$, given by 
\begin{align}
  \left[ \Gram \right]_{nm} &:= \SP{p(\cdot-m)}{p(\cdot -n)} \quad\text{for}\quad n,m\in\Z \label{eq:Graminf},
\end{align}
can be defined as $\GramM:=\vPM \Gram \vPM$, see \cite[Prop. 5.1.5]{Gro01}.
If $p$ satisfies \eqref{eq:rieszconditionorig} and if we restrict the semi-discrete convolution $p \convdis \vc$ to
$\ell^2_M$, we obtain a $2M+1$ dimensional Hilbert subspace $\SI^M (p)$ of $\SI(p)=\Hil$.  Then the unique linear
operation $B^M$, which generates from $\{p(\cdot -m)\}_{m=-M}^M$ an \emph{orthonormal basis}
(ONB) $\{\poM_m\}_{m=-M}^M$ for $\SI^M (p)$ and simultaneously minimizes
\begin{align}
  \sum_{m=-M}^{M} \Norm{B^M p(\cdot -m) -p(\cdot -m)}_{L^2}^2      \label{eq:mep:finite}
\end{align}
is given by the (symmetric) Löwdin transformation \cite{Low50,Loe70,Lan36,Jun07}
  \begin{mycomment}
  \footnote{Note that in Physics, the first input of the inner product is complex
  conjugated, so that the Gram matrix is then defined as $[\Gram]_{nm}=(f_n,f_m)_{Physics}$.}
  \end{mycomment}
and can be represented in matrix form as

\begin{align}
  \poM_m := B^M p(\cdot -m) =\sum_{n=-M}^{M} [\GramMinv ]_{mn} p(\cdot -n) \quad\text{for all}\quad  m\in\{-M,\dots,M\}
  ,\label{eq:co}
\end{align}
where we call each $\poM_m$ a \emph{Löwdin orthogonal (LO) pulse} or \emph{Löwdin pulse}.

Here $\GramMinv$ denotes the (canonical)  inverse square-root (restricted to $\ell^2_M$) of
  $\GramM$.  Note that $\GramMinv$ is not equal to $\vPM \Graminv\vPM$.
Since the sum in \eqref{eq:co} is finite, the definition of the Löwdin pulses is also pointwise well-defined. In the
next section, we will see that this is a priori not true for the infinite case.  If we identify the corresponding $m$'th
row of the inverse square-root of $\GramM$ with vectors $\vh^{M}_{m}=([\GramMinv]_{m \,-\!M}, \dots,[\GramMinv]_{m \,
+\! M})$ we can describe \eqref{eq:co} by a FIR filter bank as 
\begin{align} 
  \poMm = p \convdis \vh^{M}_{m}\quad,\quad m\in\{-M,\dots,M\}.\label{eq:lofilter}
\end{align}

Unfortunately, none of these Löwdin pulses is a shift-orthonormal pulse, which would be necessary for an OOPPM
transmission. In the next section we will thus show that the Löwdin orthogonalization converges for $M$ to infinity to
an IIR filter ${\vbb}$ given as the centered row of $\Graminv$. This IIR filter generates then a \changes
shift-orthonormal \changee pulse, namely the \changes square-root \changee Nyquist pulse defined in \eqref{eq:nyquist}.
Hence the Löwdin orthogonalization \eqref{eq:co} provides an approximation to our OOPPM design.
In the following we will investigate its stability, i.e. its convergence property.

%We will show in the next sections that for bounded stable generators $p\in L^2_K$, there exists at least one IIR filter
%$\vh_p$ given by the inverse square-root of the spectrum of the Gram matrix for $p$.  For this $\vh_p$ we will establish
%two filter approximation methods and show that they generate an optimal orthonormal generator $\po$ minimizing
%$\Norm{p-\tpo}_{L^2}$ among all orthonormal generators $\tpo$ for $\SI(p)$.
%

%%%%%%%%%%%%%%%%%%%%%%%%%%%%%%%%%%%%%%%%%%%%%%%%%%%%%%%%%%%%%%%%%%%%%%%%%%%%%%%%%%%%%%%%%%%%%%%%%%%%%%%%%
\subsection{Stability and Approximation}\label{sec:loinfinite}
%%%%%%%%%%%%%%%%%%%%%%%%%%%%%%%%%%%%%%%%%%%%%%%%%%%%%%%%%%%%%%%%%%%%%%%%%%%%%%%%%%%%%%%%%%%%%%%%%%%%%%%%%

In this section we investigate the limit of the Löwdin orthogonalization in \eqref{eq:co} for translates (time-shifts)
of the optimized pulse $p$ with time duration $T_p <\infty$  where we further assume that $p$ is a bounded stable
generator.  If we set $K:=\lfloor T_p\rfloor$ then certainly $p\in L^2_K$.  In this case the auto--correlation of $p$ 
\begin{align}
  r_p(t):=(p*\cc{p}_{-})(t)= \int_\R p(\tau) \cc{p(\tau-t)} d\tau\label{eq:autocor},
\end{align}
with the time reversal $p_{-}(t):= p(-t)$ is a compactly supported bounded function on $[-K,K]$.  Due to the Poisson
summation formula we can represent $\Phi_p$ almost everywhere by the Fourier series \eqref{eq:dft} ($T_0=1$) of the samples
$\{r_p(n)\}$
\begin{align}
  \Phi_p(\nu)= \sum_{n=-\infty}^\infty  r_p (n) e^{-2\pi i n \nu} 
  = \sum_{n=-K}^K [\Gram]_{n0}e^{-2\pi i n \nu} 
   \label{eq:symbolpoisson},
\end{align}
which is the symbol of the Toeplitz matrix $\Gram$, since we have from \eqref{eq:Graminf} and \eqref{eq:autocor}
that $r_p(n-m)=[\Gram]_{nm}$. Moreover the symbol is continuous since the sum is finite due to the compactness of $r_p$.

On the other hand the initial pulse $p$ is a Wiener function\footnote{
Wiener functions are locally bounded in $L^\infty$ and globally in $\ell^1$.  }
\cite[Def. 6.1.1]{Gro01}
so that $\Phi_p$ defines a continuous function  and condition \eqref{eq:rieszcondition} holds pointwise
\cite[Prop.1]{AST01}. Since both sides in \eqref{eq:symbolpoisson} are identical a.e.  they are identical everywhere by
continuity (see also \cite[p.105]{Gro01}). Thus, the spectrum of $\Gram$ is continuous, strictly positive and bounded by
the Riesz bounds.  Hence the inverse square-roots of $\Gram$ and $\GramM$ exists s.t. for any $M\in\N$ (by Cauchy's
interlace theorem, \cite[Th.  9.19]{BG05})
\begin{align}
  A\eins_M \leq \Gram_M \leq B\eins_M \quad\text{and}\quad \frac{1}{B} \eins_M \leq \GramMi \leq \frac{1}{A}
        \eins_M\label{eq:finitebound},
\end{align}
where $\eins_M$ denotes the identity on $\ell^2_M$. Now we can  approximate the Gram matrix by \namen{Strang}'s
circulant preconditioner \cite{Str86}, s.t. the diagonalization is given by a discrete Fourier transform (DFT)
\cite{Dav79}.  To get a continuous formulation of %
the approximated Löwdin pulses we use the \namen{Zak} transform \cite{Jan88}, given for a continuous function $f$ as 
\begin{align}
  (\Zak f) (t,\nu) := \sum_{n\in\Z} f(t-n) e^{2\pi i n\nu}\label{eq:zak}\quad\text{for}\quad t,\nu \in\R.
\end{align}
Our main result is the following theorem.
%%%%%%%%%%%%%%%%%%%%%%%%%%%%%%%%%%%%%%%%%%%%%%%%%%%%%%%%%%%%%%%%%%%%%%%%%%%%%%%%%%%%%%%%%%%%%%%%%%%%%%%%%%%%%%%%%%%%%%%%
%%%%%%%%%% Neu überarbeitet am 15.01.2010
%-------------------------------------------------
%%% THEOREM:MAIN %%%%%%%%%%%%%%%%%%%%%%%%%%%%%%%%%%%
\begin{thm}\label{th:lozak}
  Let $K\in\N$ and  $p\in L^2_K$ be a continuous stable generator for $\SI(p)$.  Then we can approximate the limit set
  of the Löwdin pulses $\{\po_m\}$ by a sequence of finite function sets $\{\tp_m^{\circ,M}\}_{m=-M}^M$, which are
  approximate Löwdin orthogonal (ALO). The functions $\tp_m^{\circ,M}$ are given pointwise for $M\geq K$ and
  $m\in\{-M,\dots,M\}$ by the Zak transform as
  \begin{align}
      \tpoMm\! (t)\! :=\! \begin{cases}
        \frac{1}{2M\!+\!1} \overset{2M}{\underset{{l=0}}{\sum}}
        \frac{e^{\frac{-2\pi i ml}{2M+1}}(\Zak p)(t,\frac{l}{2M+1})}
          {\sqrt{(\Zak r_p)(0,\frac{l}{2M+1})}}
              & \Betrag{t}\!\leq\! M\!-\!\frac{K}{2}\\
          0 &\text{else}\end{cases},\label{eq:zakpo}
  \end{align}
  such that for each $m\in\Z$
  \begin{align}
      \tpo_m (t) =\lim_{M\to\infty}\tpoMm (t) \label{eq:zakpolimit}
  \end{align}
  converges pointwise on compact sets. The limit in \eqref{eq:zakpolimit} can be stated as 
  \begin{align}
    \htpo (\nu) &= \hp (\nu)\cdot \left(\Phi_p(\nu)\right)^{-\frac{1}{2}}  \label{eq:lofreq}
  \end{align}
  for $\nu\in\R$ in the frequency-domain. Hence the Löwdin generator $\tpo:=\tpo_0$ is an orthonormal generator for
  $\SI(p)$.
\end{thm}
%%%%%%%%%%%%%%%%%%%%%%%%%%%%%%%%%%%%%%%%%%%%%%%

%%%%%%%%%%%%%%%%%%%%%%%%%%% PROOFs %%%%%%%%%%%%%%%%%%%%%%%%%%%%%%%%%%%%%%%%%%%%%%%
\begin{proof}

The proof consists of two parts. In the first part we derive an straightforward finite construction in the time domain
to obtain time-limited pulses \eqref{eq:zakpo} being approximations to the Löwdin pulses. Using Strang's circulant
preconditioner the ALO pulses can be easily derived in terms of DFTs. In the second part we will then show that this
finite construction is indeed a stable approximation to the \changes square-root \changee Nyquist pulse. Here we need
pointwise convergence, i.e.  convergence in $\ell^\infty$ (the set of bounded sequences). Finally, to establish the
shift-orthogonality we use properties of the Zak transform.

Since the inverse square-root of a $N\times N$ Toeplitz matrix is hard to compute, we approximate for any $M\geq K$
 the Gram matrix $\GramM$ by using Strang's circulant preconditioner $\tGramM$
  \cite{Str86,CJ07}.  Moreover, the Gram matrix is hermitian and banded such that we can define the elements of the
  first row by \cite[(4.19)]{Gra06} as
\begin{align}
    [\tGramM]_{0n}:= \begin{cases}
        r_p(n) & n\in \{0,\dots,K\}\\
        r_p(N-n) & n\in\{N-K,\dots, N-1\}\\
        0 & \text{else}
    \end{cases}.\label{eq:strang}
\end{align}
Here we abbreviate $N:=2M+1$.  The crucial property of Strang's preconditioner $\tGramM$ is the fact that the
eigenvalues $\lam_l(\tGramM)$ are sample values of the symbol $\Phi_p$ in \eqref{eq:symbol}. This special property is in
general not valid for
other circulant preconditioners \cite{CJ07}. To see this, we derive the eigenvalues by
\cite[Theorem~7]{Gra06} as
\begin{align}    
  \tlam_l^{M}:=\lam_l(\tGramM)&=\sum_{n=0}^K r_p(n) e^{-2\pi i n\frac{l}{N}} + \sum_{n=N-K}^{N-1} r_p (N-n) e^{-2\pi i
  n\frac{l}{N}  } \quad \text{for} \quad l \in\{0,\dots,2M\}\label{eq:eigenvaluedecomp} 
\end{align}
by inserting the first row of $\tGramM$ given in \eqref{eq:strang}. If we set in the second sum $n'=n-N$ we get from
\eqref{eq:symbolpoisson} 
\begin{align}
  \tlam^{M}_l& =\sum_{n=-K}^K r_p (n) e^{-2\pi i l \frac{n}{N}} = \Phi_p
  \left(\frac{l}{2M+1}\right)\label{eq:eigenvaluecp}.
\end{align}
Since $p$ is compactly supported the symbol $\Phi_p$ is continuous and the second equality in \eqref{eq:eigenvaluecp}
holds pointwise.  Moreover, the Riesz bounds \eqref{eq:rieszcondition} of $p$ guarantee that $\tGramM$ is strictly
positive and invertible for any $M$. Now we are able to define the ALO pulses in matrix notation by setting\footnote{By
slight abuse of our notation we understand in this section any matrix as an $N\times N$ matrix and $\vpM(t),\tvpoM(t)$
as $N$-dim. vectors for any $t\in\R$.}\label{footabuse} in
$\vpM(t):=\{p(t\!-\!n)\}_{n=-M}^M=(p(t\!+\!M),\dots,p(t\!-\!M))^T$ for any $t\in\R$
\begin{align}
  \tvpoM(t):=  \tGramMinv \vpM(t)&= \FmatrixM \tDiagMinv \FmatrixMa \vp^M (t),\label{eq:alo}
\end{align}

since the circulant matrix $\tGramM = \FmatrixM \tDiagM \FmatrixMa$ can be written by the unitary
$N\times N$ DFT  matrix $\FmatrixM$, with
\begin{align}
  &[\FmatrixM ]_{nm} :=\frac{1}{\sqrt{N}} e^{-2\pi i \frac{nm}{N}}\quad\text{with}\quad n,m \!\in\!\{0,\dots,2M\} 
\end{align}

and the diagonal matrix $\tDiagM$ of the eigenvalues of $\tGramM$.

Let us start in \eqref{eq:alo} from the right by applying the IDFT matrix $\FmatrixMa$,  then we get for
any $k$th component with $k\in\{0,\dots,2M\}$
\begin{align}
  [\FmatrixMa\vp^M (t)]_k
&=\frac{1}{\sqrt{N}} \Bigg(\sum_{n=0}^M  e^{2\pi i  \cdot \frac{n}{N} k} \cdot p(t+M-n) 
  %& \phantom{kkkkkkkkkk} + \sum_{j=M+1}^{2M} e^{2\pi i \cdot \frac{j}{N} k} \cdot p(t-j+M) \Bigg)\notag.\\
  + \sum_{n=M+1}^{2M} e^{2\pi i \cdot \frac{n}{N} k} \cdot p(t+M-n) \Bigg)\label{eq:idftapp}\\
  &\!\!\! \overset{\overset{j=n-M}{\downarrow}}{=}
\frac{1}{\sqrt{N}} \sum_{j=-M}^{M} e^{2\pi i \cdot \frac{j+M}{N} k} \cdot p(t-j)\label{eq:idftapp2}.
\end{align}
Next we multiply with the components 
$[\tDiagMinv]_{kl}=\del_{kl}/\sqrt{\tlam^{M}_l}$ of the inverse square-root of the diagonal matrix $\tDiagM$
\begin{align}
    \left[ \tDiagMinv \FmatrixMa\vp^M (t) \right]_l &=\frac{1}{\sqrt{N}} \sum_{j,k} \left(
    \frac{\del_{kl}}{\sqrt{\tlam_l^{M}}} e^{2 \pi i \frac{j+M}{N} k} \cdot p(t-j) \right) = \frac{1}{\sqrt{N }} \left( \sum_{j=-M}^M \frac{e^{2 \pi i \frac{j+M}{N}l}  \cdot p(t-j)}{\sqrt{\tlam^M_l}}
    \right).
\end{align}
In the last step we evaluate the DFT at $m\in\{0,\dots,2M\}$
\begin{align}
  \tpoMnM (t)=[\tvpoM (t)]_m 
  &= \frac{1}{N} \sum_{l=0}^{2M}
  \left(  \frac{e^{-2\pi i m \cdot\frac{l}{N}}}{\sqrt{\tlam_l^M}} \sum_{j=-M}^M p(t-j) e^{2 \pi i l \frac{j+M}{N}}
  \right) = \frac{1}{N} \sum_{l=0}^{2M}  e^{-2\pi i l\frac{m-M}{N}} \frac{\sum_{j=-M}^M p(t-j) e^{2\pi i l \cdot
  \frac{j}{N}}}{\sqrt{(\Zak r_p)(0,\frac{l}{N})}}.\label{eq:lastfinite}
\end{align}
where we used the Zak transform \eqref{eq:zak} of $r_p$ to express the eigenvalues $\tlam^M_l$. %
In the next step we extend the  DFT sum of the numerator in \eqref{eq:lastfinite} to an infinite sum. This is possible
since $p(\cdot-k)$ always has the same support length $K$ for each $k\in \Z$.  Thus, for all $\Betrag{t}>M-\frac{K}{2}$
the non-zero sample values are shifted in the kern of $\vP_M$; hence $\vpM(t)=0$. On the other hand for $\Betrag{t}\leq
M-\frac{K}{2}$ any shift $\Betrag{j}>M$ results in $p(t-j)=0$.  If we also set for each $M$ the index
  $k:=m-M$ in \eqref{eq:lastfinite} then the continuous ALO pulses can be written as
\begin{align}
  \tpoMk (t)= 
    \begin{cases}
        \frac{1}{N} \sum_l e^{-2\pi i l\frac{k}{N}} \frac{(\Zak p)(t,\frac{l}{N})}{\sqrt{(\Zak r_p)(0,\frac{l}{N})}} &
        \Betrag{t}\leq M\!-\!\frac{K}{2} \\
        0 & \text{else}
    \end{cases}.\label{eq:apporthotrafo}
\end{align}
This defines for each $k$ the operation $\tilde{B}_k^M$ by $\tpoMk = \tilde{B }_k^{M} p$.
For each $t\in\R$ the term  $\frac{(\Zak p)(t,\nu)}{\sqrt{(\Zak \tilde{r})(0,\nu})}$ is  a continuous function
in $\nu$ since the Zak transforms are finite sums of continuous functions and the nominator vanishes nowhere. This is
guaranteed by the positivity and continuity of $\Phi_p$ due to \eqref{eq:finitebound}. Hence the ALO pulses are
continuous as well.

 The second part of the proof shows the convergence of our finite construction to a \changes square-root \changee
 Nyquist pulse.  Therefore we use the \emph{finite section method} for the Gram matrix.  Gray showed \cite[Lemma
 7]{Gra06} that
\begin{align}
  \Norm{\tGramM - \GramM}_{w} \to 0
\end{align}
as $M \to \infty$ in the weak norm $\Norm{A}^2_{w}:=1/N \sum_j \lambda_j^2(A)$ implying  weak convergence of the
operators.  Since  $\tGramM$ is strictly positive for each $M\in\N$ we get by \cite{GC08} 
\begin{align}
  \Norm{\tGramMinv - \GramMinv }_{w} \to 0.\label{eq:ae}
\end{align}
Unfortunately this does not provides a strong convergence, which is necessary to state convergence in $\ell^2$ %
\begin{align}
   \tGramMinv \vP_M \vc \to \Graminv \vc \quad \text{for any } \vc\in\ell^2   \label{eq:tgramconv}.
\end{align}    
However, from \cite{SJB03} \emph{finite strong convergence} can be ensured, i.e. convergence of \eqref{eq:tgramconv} for
all $\vc\in\ell^2_{M'}$ for each $M'\in\N$. But for any $t\in\R$ 
 there exists an $M'$ sufficiently large, due to the compact support property of $p$, such that
 $\vc=\vp(t):=\{p(t-n)\}\in\ell^2_{M'}$. This is in fact sufficient,
since it implies pointwise convergence in $\ell^\infty_{M'}$ of \eqref{eq:tgramconv}, i.e. component-wise
convergence for each $t\in\R$.  Let us take for each  $t\in\R$ the number $M'\in\N$ such that
$\max\{\Betrag{t},K\}\leq M'$. Then we can define the limit of the $k$th component  as 
\begin{align}
  \tpo_k (t) := \lim_{M\to \infty}  \tpoMk (t) .\label{eq:polimit}
\end{align}
If we define $\Del \nu= \frac{1}{2M+1}$ and $\nu_l=l \Del\nu$  we can write for \eqref{eq:polimit} by
inserting \eqref{eq:apporthotrafo} %
\begin{align}
  \lim_{M\to\infty} \tpoMk (t) &= \lim_{M\to\infty} \sum_{l=0}^{2M} \frac{e^{-2\pi i k \nu_l}(\Zak p)(t,\nu_l)}{\sqrt{(\Zak
  r_p)(0,\nu_l)}} \Del \nu.\label{eq:limitZak}
\end{align}
Using the quasi-periodicity \cite[(2.18),(2.19)]{Jan88} of the Zak transform for $k\in\Z$  we have for any
$t,\nu\in\R$ 
\begin{align}
    (\Zak p(\cdot -k) )(t,\nu) &= (\Zak p)(t-k,\nu) =e^{-2\pi i \nu k} (\Zak p)(t,\nu)
  \label{eq:zakshift}.
\end{align}
We can express the partial sum on the right hand side of \eqref{eq:limitZak} in the limit as a Riemann integral for each $t\in\R$ 
\begin{align}
  \tpo_k (t) &:=\lim_{M\to\infty} \tpoMk (t) = \int_0^1 \frac{(\Zak p(\cdot-k)) (t,\nu)}{\sqrt{(\Zak r_p)(0,\nu)}} d\nu.
  \label{eq:ortholimit}
\end{align}
This shows that $\tilde{B}^{M}_k p$ converge pointwise for $M\to\infty$ to $\tilde{B}_k p = \tilde{B_0} p(\cdot
-k)=\tpo_k$ for each $k$.
The sequence $\{\tpo_k\}$
is then generated by shifts \changes of the centered pulse \changee $\tpo:=\tpo_0 $ since
the shift operation commutes with $\tilde{B}:=\tilde{B_0}$. This in turn commutes with the Zak
transformation,
i.e. for all $t\in\R$ we have
\begin{align}
    (\tilde{B}p)(t\!-\!k) &= \int \frac{(\Zak p)(t\!-\!k,\nu)}{\sqrt{(\Zak r_p)(0,\nu)}}d\nu = \int
    \frac{(\Zak p(\cdot\! -\!k))(t,\nu)}{\sqrt{(\Zak r_p)(0,\nu)}}d\nu \\
  &=  (\tilde{B} p(\cdot\!-\!k)) (t).
\end{align}
From \eqref{eq:ortholimit} it is now easy to show that $\tpo$ is an orthonormal generator. We write the left
hand side of \eqref{eq:ortholimit} in the Zak domain, by applying the Zak transformation%
\footnote{A similar result is also known in the context of Gabor frames, see also \cite[8.3]{Gro01}.}
to $\tpo$
\begin{align}
    (\Zak \tpo) (t,\nu) &=  \frac{(\Zak p) (t,\nu)}{\sqrt{(\Zak r_p)(0,\nu)}} \label{eq:cozak}.
\end{align}
If we multiple \eqref{eq:cozak} by the exponential and integrate over the time we yield for every $\nu\in\R$
\begin{align}
  \int_0^1 e^{-2\pi i \nu t} (\Zak \tpo) (t,\nu) dt &=  \int_0^1  \frac{e^{-2\pi i \nu t}\cdot (\Zak
  p)(t,\nu)}{\sqrt{\Zak r_p (0,\nu)}} dt.\label{eq:zakinverse}
\end{align}
Since $\Phi_p=(\Zak r_p )(0,\cdot)$ is time-independent we get the ``orthogonalization trick'' \eqref{eq:nyquist} by
using in \eqref{eq:zakinverse} the inversion formula \cite[(2.30)]{Jan88} of the Zak transform
\begin{align}
  \htpo (\nu) = &\hp (\nu) \cdot \left(\Phi_p(\nu)\right)^{-\frac{1}{2}}=\hpo(\nu).
\end{align}
Again, this is also defined pointwise  since the right hand side is continuous in $\nu$. It can now be easily verified
that $\tpo$ fulfills the \changes shift-orthonormal \changee condition \eqref{eq:orthocondition}, which shows that
$\tpo$ is an orthonormal generator for $\SI(p)$.
\end{proof}
\begin{bemerkung}
  Note, that relation \eqref{eq:zakshift} induces a time-shift. To apply this to the ALO pulses in
  \eqref{eq:apporthotrafo} the time domain has to be restricted further. Hence the ALO pulses do not have global shift
  character for finite $M\in\N$, but locally, i.e. $\tpoMk$ shifted back to the
  center matches $\tpoM$ for $t\in[-M+\frac{K}{2}+\Betrag{k},M-\frac{K}{2}-\Betrag{k}]$:
  \begin{align}
  \tpoMk (t+k) &= \frac{1}{N} \sum_l e^{-2\pi i \frac{l}{N}k} \frac{\sum_{n=-M}^M p(t+k-n) e^{2\pi i
  \frac{l}{N}n}}{\sqrt{\tlam^M_l}}=\frac{1}{N} \sum_l  
        \frac{\sum_{n=-M}^M p(t+k-n) e^{2\pi i \frac{l}{N}(n-k)}}{\sqrt{\tlam^M_l}}\\
      &=\frac{1}{N} \sum_l  \frac{\sum_{n'=-M-k}^{M-k} p(t-n') e^{2\pi i \frac{l}{N}n'}}{\sqrt{\tlam^M_l}}.
      \intertext{Since $p(t+M+\Betrag{k})=0$ and $p(t-M-\Betrag{k})=0$ for $\Betrag{k}<M$ and $\Betrag{t}\leq
      M-\frac{K}{2} -\Betrag{k}$, we end up with}
    \tpoMk (t+k)&=\frac{1}{N} \sum_l  \frac{\sum_{n=-M}^{M} p(t-n) e^{2\pi i \frac{l}{N}n}}{\sqrt{\tlam^M_l}}=\tpoM(t).
    \end{align}
    For all $\Betrag{k}<M$ the ALO pulses have the same shape in the window $\Betrag{t}\leq M-\frac{K}{2}-\Betrag{k}$ if
    we shift them back to the origin. 
    %In Section \ref{sec:TFandONB} we will show that the Löwdin generator $\tpo$ is the canonical tight frame $\po$,
    %i.e. $\tpo (t) = \po(t)$ for all $t\in\R$.
\end{bemerkung}

    Moreover, the ALO pulses are all continuous on the real line, since  they are a finite sum of continuous functions
    by definition \eqref{eq:alo}. Hence each ALO pulse goes continuously to zero at the support boundaries.  So far it
    is not clear whenever $\tpo$ is continuous or not. Nevertheless its spectrum $\htpo$ is continuous and so we can
    state $\tpo=\po$ almost everywhere. Hence the orthogonalization trick defines the \changes square-root \changee
    Nyquist pulse $\po$ only in an $L^2$ sense. 
%
%%%%%%%%%%%%%%%%%%%%%%%%%%%%%%%%%%%%%%%%%%%%%%%%%%%%%%%%%%%%%%%%%%%%%%%%%%%%%%%%%%%%%%%%%%%%%%%%%%%%%%%%%
\changes\section{Discussion of the Analysis}\label{sec:discussion}\changee
%%%%%%%%%%%%%%%%%%%%%%%%%%%%%%%%%%%%%%%%%%%%%%%%%%%%%%%%%%%%%%%%%%%%%%%%%%%%%%%%%%%%%%%%%%%%%%%%%%%%%%%%%

 In this section we will discuss now the properties of our OOPPM design for UWB, i.e.  the optimization and
orthogonalization, which can be completely described by an IIR filtering process.  First we will relate the Löwdin
orthogonalization to the canonical tight frame construction.  Afterwards we will show in \secref{sec:optimality} that
the Löwdin transform yields the orthogonal generator with the minimal $L^2$-difference to the initial optimized pulse.
This is the same optimality property as for canonical tight frames \cite{JS02}.  But such an energy optimality does not
guarantee FCC compliance. So we will discuss in \secref{sec:interdependence} the influence of a perfect orthogonalization
to the FCC optimization.  Finally, we will discuss the implementation of a perfect orthogonalization by FIR
filtering.

%filter applied to a fixed basic pulse, we shall give now a general characterization of this filter process.
%This concerns the localization properties in time and frequency as well as optimality in energy. It is well known that
%the Löwdin transform is the optimal orthonormalization for minimizing \eqref{eq:mep:finite}, i.e. in an $L^2$-sense. But
%our attempt is to consider simultaneous $L^\infty$-control which usually requires advanced Banach space techniques.
%Additionally, the $L^\infty$ control is needed locally, see the application of the design in
%\secref{sec:app}. 

%%%%%%%%%%%%%%%%%%%%%%%%%%%%%%%%%%%%%%%%%%%%%%%%%%%%%%%%%%%%%%%%%%%%%%%%%%%%%%%%%%%%%%%%%%%%%%%%%%%%%%%%%
\subsection{Relation Between Tight Frames and ONBs}\label{sec:TFandONB}
%%%%%%%%%%%%%%%%%%%%%%%%%%%%%%%%%%%%%%%%%%%%%%%%%%%%%%%%%%%%%%%%%%%%%%%%%%%%%%%%%%%%%%%%%%%%%%%%%%%%%%%%%
%
%It is known, that the frame operator and the shift operator commutes, \cite[Lemma 7.3.7]{Chr03}.  
 Any Riesz basis $\{p_k\}$ for a Hilbert space $\Hil$ is also a \emph{exact frame} for $\Hil$ with the frame
operator $S$ defined by
\begin{align}
    S: \Hil \to \Hil,\quad 
    f \mapsto Sf= \sum_k \SPH{f}{p_k} p_k,\label{eq:frame}
\end{align}
where the frame bounds are given by the Riesz bounds $0\!<\!A\leq B\!<\!\infty$ of $\{p_k\}$ \cite[Th.  5.4.1,
6.1.1]{Chr03}, i.e. 
  \begin{align}
    A\Norm{f}_{_{\Hil}}^2\leq \SPH{Sf}{f}\leq B\Norm{f}_{_{\Hil}}^2\quad\text{for any}\quad f\in\Hil.
  \end{align}
  Here $\SPH{\cdot}{\cdot}$ denotes the inner product in $\Hil$ and $\Norm{\cdot}_{\Hil}$ the induced norm. Since $S$ is
  bounded and invertible, i.e. the inverse operator exists and is bounded \cite{Chr03}, we can write each $f\in\Hil$ as
\begin{align}
  f=SS^{-1}f= \sum_k \SPH{S^{-1} f}{p_k} p_k \label{eq:frameinv}.
\end{align}
In this case the Löwdin orthonormalization corresponds to the \emph{canonical tight frame construction}.

%In this section we will show rigorously that the Löwdin orthonormalization of a Riesz--basis for 
%$\Hil\subset L^2$ is equivalent to the canonical tight frame construction of a so-called \emph{exact frame} (see below). 
%The frame operator $S$ for any
%frame $\{p_k\}$ given by 
%
%

%
%%%%%%%%%%%%%%% LEMMA: Meyer %%%%%%%%%%%%%%%%%%%%%%%%%%%%%%%%%%
\begin{lemma}\label{le:meyer}
  Let the sequence $\{p_k\}$ be a Riesz basis for the Hilbert space $\Hil:=\cc{\spann\{p_k\}}$ and $\Gram$ its Gram
  matrix. Then the canonical tight
  frame $\{\po_k\}$ is given for each $k\in\Z$ by:   
  \begin{align}
      \po_k := S^{-\frac{1}{2}} p_k = \sum_l [\Graminv]_{kl}  \ p_l \label{eq:framegram}
  \end{align}
  in an $L^2$-sense.\footnote{ This statement was already given without further explanation by Y. Meyer in \cite{Mey86}
  equation (3.3). Note that Y. Meyer used condition (3.1) and (3.2) in \cite{Mey86} which are equivalent to the Riesz basis
  condition.}
  %Here, $\Graminv$ is the inverse square root of the Gram matrix of the Riesz basis $\{p_k \}$  %
\end{lemma}
\begin{paragraph}{Proof}
  See \ref{app:meyer}.
\end{paragraph}
%
%
%%%%%%%%%%%%%%%%%%%%%%%%%%%%%%%%%%%%%%%%%%%%%%%%%%%%%%%%%%%%%%%%%%%%%%%%%%%%%

%\begin{bemerkung} Here we denote with $S^{-\frac{1}{2}}$ resp. $\Graminv$ the canonical inverse square root, i.e.
%operators with real positive eigenvalues. The matrix representations of $\Gram$ and $\mC$ are with respect to the Riesz
%basis $\{p_k\}$. For non-real or negative square-roots of the eigenvalues see the discussion of the optimality in the
%next section.  \end{bemerkung}
%
If we now set $p_n := p(\cdot -n)\in L^2$ the Riesz-basis is generated by shifts of a stable generator and $\Hil=\SI(p)$
becomes a 
principal shift-invariant (PSI) space,

which is a separable Hilbert subspace of $L^2$ as discussed in Section \ref{sec:shiftspaces}.  The canonical tight frame
construction then generates a shift-orthonormal basis, i.e. an orthonormal generator.  The reason is that
shift-invariant frames and Riesz bases are the same in regular shift-invariant spaces \cite[Th.2.4]{CCK01}. So any frame
becomes a Riesz basis (exact frame) and any tight frame an ONB (exact tight frame). Hence for regular PSI spaces there
exists no redundancy for frames. This generalize the Löwdin transform for generating a \changes square-root \changee
Nyquist pulse to any stable generator $p$.

From Meyer \cite{Mey86} we know that \eqref{eq:framegram} can be written in frequency domain as the
orthogonalization trick. Therefore the limit of the Löwdin transformation
$\tilde{B}=B\colon \SI(p) \to \SI(p)$
\begin{align}
  f&\mapsto Bf=\int_0^1 \frac{(\Zak f)(\cdot,\nu)}{\sqrt{\Zak r_p)(0,\nu)}}d\nu
\end{align}
equals the inverse square-root of the frame operator in
\eqref{eq:framegram}.

%%%%%%%%%%%%%%%%%%%%%%%%%%%%%%%%%%%%%%%%%%%%%%%%%%%%%%%%%%%%%%%%%%%%%%%%%%%%%%%%%%%%%%%%%%%%%%%%%%%%%%%%%
\subsection{Optimality of the Löwdin Orthogonalization}\label{sec:optimality}
%%%%%%%%%%%%%%%%%%%%%%%%%%%%%%%%%%%%%%%%%%%%%%%%%%%%%%%%%%%%%%%%%%%%%%%%%%%%%%%%%%%%%%%%%%%%%%%%%%%%%%%%%
%
Janssen and Strohmer have shown in \cite{JS02} that the canonical tight-frame construction of Gabor frames for $L^2$ is
via Ron-Shen duality equivalent to an ONB construction on the adjoint time-frequency lattice. Furthermore they
have shown that among all tight Gabor frames, the canonical construction yields this  particular generator with minimal
$L^2$-distance to the original one. However, for SI spaces this optimality of the Löwdin orthogonalization has to be
proved otherwise. To prove this we use the structure of regular PSI spaces.
%
%%% THEOREM %%%%%%%%%%%%%%%%%%%%%%%%%%%%%%%%%%%%%%%%%%%%%%%%%%%%%%%%%%%%%%%%
\begin{thm}\label{th:unique}
  The unique orthonormal generator with the minimal $L^2$
  distance to the normalized stable generator $p\in L^2$ for $\SI(p)$ is given by the Löwdin generator $\po$.
\end{thm}
%%%%%%%%%%%%%%%%%%%%%%%%%%%%%%%%%%%%%%%%%%%%%%%%%%%%%%%%%%%%%%%%%%%%%%%%%%%%
%
\begin{paragraph}{Proof}
  See \ref{app:proofunique}.
\end{paragraph}

Nevertheless, We have to rescale the orthonormal generator $\po$ to respect the FCC mask, see \secref{sec:app}. For this the maximal
difference of the power spectrum\footnote{In fact the $L^\infty$-distance of the FCC mask $\SFCC$ and $\Betrag{\hpo}^2$
in $F$ is relevant, assumed $\Betrag{\hpo}^2$ is bounded by $\SFCC$.} of the (normalized) optimal designed pulse and the
orthonormalized pulse is of interest, i.e.
\begin{align}
  \Norm{\Betrag{\hp}^2-\Betrag{\hpo}^2}_{L^\infty}= \esssup{\nu\in\R}{ \Betrag{\frac{\Phi_p (\nu) -1}{\Phi_p
  (\nu)}}\cdot\Betrag{\hp(\nu)}^2}  \leq \Norm{\frac{\Phi_p
  -1}{\Phi_p}}_{L^\infty} \cdot \Norm{\hp}_{L^\infty}^2 .
\end{align}
This shows again that this $L^\infty$ distortion is also determined by the spectral properties of the optimal designed
pulse $p$ and its Riesz bounds. Unfortunately it is very hard to control the optimization and orthogonalization filter
simultaneously as will be shown in the next section.

%%%%%%%%%%%%%%%%%%%%%%%%%%%%%%%%%%%%%%%%%%%%%%%%%%%%%%%%%%%%%%%%%%%%%%%%%%%%%%%%%%%%%%%%%%%%%%%%%%%%%%%%%
\subsection{Interdependence of Orthogonalization and Optimization}\label{sec:interdependence}
%%%%%%%%%%%%%%%%%%%%%%%%%%%%%%%%%%%%%%%%%%%%%%%%%%%%%%%%%%%%%%%%%%%%%%%%%%%%%%%%%%%%%%%%%%%%%%%%%%%%%%%%%
%
The causal FIR \changes operation \changee in \eqref{eq:firfilter} of a fixed initial pulse $q$ of odd order $L$  with clock
rate $1/T_0$ can be also written in the time-symmetric form as a real semi-discrete convolution  
\begin{align}
  p = q \convdisTo \vgL \quad\text{for}\quad \vgL \in \changes \ell^2_{\frac{L-1}{2}}(\R).\changee
\end{align}
In this section we investigate the interdependence of the IIR filter $\vh$ and the FIR filter
$\vgL$, i.e. the interdependence of the orthogonalization filter in \secref{sec:lofinite} and the FCC optimization
filter in \secref{sec:pulseshapedesign}  for different clock rates. So far we have first optimized spectrally and
afterwards performed the orthonormalization. In this order for a chosen $q$, the orthogonalization filter $\vh$ depends on
$\vgL$, hence we write $\vh=\vh_g$. Moreover the clock rates of the filters differ, hence we stick the time-shifts as
index in the semi-discrete convolutions. For the $T$-shift-orthogonal pulse we get then
\begin{align}
  \pTo = (q  \convdisTo \vgL ) \convdisT \vh_g \label{eq:hgconv}.
\end{align}
Let us set $T=\Delta T_0$ and $T_q= N_q T_0$ for $N_q \in\N, \Delta>0$. 
Since the filter clock rate of $\hvg$ is fixed to $1/T_0$ to ensure full FCC--range control, the variation is expressed
in $\Delta$. To get rid of $T_0$ we scale the time $t$ to $t'=t/{T_0}$ such that the time--shift of $\vgL$ is $T'_0=1$.
\changes
\begin{comment}
  Then \eqref{eq:hgconv} becomes
\begin{align}
  \poD = (q \convdis \vgL ) \convdisD \vh_g \label{eq:hgconvD}
  &= p \convdisD \vh_g.
\end{align}
\end{comment}
\changes We \changee observe the following effects:
  1) If $\frac{1}{\Delta}\in\N$ then $\poD=\qoD$.
  \label{enu:h-eats-ga} 
  2) If $\Delta\in\N$ then the distortion by 
    $\hvh_g$ is limited periodically to the interval $[-\frac{1}{2\Delta},\frac{1}{2\Delta}]$. 
    \label{enu:h-free-g}
3) If $q$ is already $1$-shift-orthogonal and $\Delta\in\N$ then $\vh_g$ can be omitted and instead adding an
    extra condition on $\vgL$ to be a $\Delta$-shift orthogonal-filter, i.e. $r_{\vgL}(k\Delta)=
    \del_{k0}$, which ensures $\Delta$-shift orthogonality of the output $p$. 
    \label{enu:filterorthogonal}
    \changee
To \changes see \changee point 1), let us first orthogonalize $q$ by $\vh_q$ and ask for the filter $\vgtDg$ which preserves the
$\Delta$-orthogonalization in the presence of $\vgL$.  Hence we aim at
\begin{align}
  \poD = p \convdisD \vh_g =\qoD \convdis \vgtDg.\label{eq:pod}
\end{align}
But from \eqref{eq:lofreq} we know how $\vh_g$ acts
in the frequency-domain:
\begin{align}
  \Betrag{\hpoD(\nu)}^2 &= \frac{\Betrag{\hp(\nu)}^2}{\frac{1}{\Delta}\sum_k \Betrag{\hp(\nu\!-\!\frac{k}{\Delta})}^2}
  \overset{\overset{\eqref{eq:pod}}{\downarrow}}{=}  \Betrag{\hqoD(\nu) \cdot \hvgtDg (\nu)}^2.\label{eq:filteradj}
\end{align}
\changes Since $\Betrag{\hp( \nu)}^2\!=\!\Betrag{\hvgL \!(\nu)\cdot \hq(\nu)}^2$ and $\frac{1}{\Delta}\!\in\!\N$ we get by
the $\frac{1}{\Delta}$-periodicity of $\Betrag{\hvgL}^2$
\begin{align}
   \Betrag{\hqoD(\nu) \cdot \hvgtDg (\nu)}^2
  =\frac{|\hvgL(\nu)|^2 \cdot\Betrag{\hq(\nu)}^2}{\frac{1}{\Delta}|\hvgL(\nu)|^2 
    \cdot\sum_k \Betrag{\hq(\nu-\frac{k}{\Delta})}^2}
   =\frac{ \Betrag{\hq(\nu)}^2}{\frac{1}{\Delta}\sum_k \Betrag{\hq(\nu-\frac{k}{\Delta})}^2}=
   \Betrag{\hqoD(\nu)}^2\label{eq:inter}.
\end{align}
\changee
\changes Hence we get $\tilde{g}_{\Delta,\vgL}(k)=\delta_{k0}$, which shows 1). The price of the \changee orthogonalization is the loss of a
frequency control, since the frequency property is now completely given by the basic pulse $q$ and \changes time-shift
\changee $\Delta$. In \figref{fig:t2t1} the effect is plotted for $\Delta \in[1,2]$ and $L=25$. For small $\Delta$ the
distortion is increase by the orthogonalization. This also shows that a perfect orthogonalization and optimization with
the same clock rates is not  possible.

In 2) a perfect orthogonalization does not completely undo the optimization, \changes since $T=\Delta>T_0=1$. \changee
For $\Delta=2$ we can describe the filter by using the addition theorem in
$\Betrag{\hvgL(\nu+1/2)}^2=\hvr(\nu+1/2)=2r_{\vg,0}-\hvr(\nu)$ by \changes
\begin{align}
  \sum_k \Betrag{\hp\left(\nu-\frac{k}{2}\right)}^2 = 
      \hvr(\nu) \Bigg[ \sum_k \Betrag{\hq\left(\nu + \frac{2k}{2}\right)}^2 +
      \frac{\hvr(\nu\!+\!\frac{1}{2})}{\hvr(\nu)} \underbrace{\sum_k \Betrag{\hq \left(\nu + \frac{2k\!+\!1}{2}\right)}^2}_{=:\Phi'_q (\nu)} \Bigg]
      =  \hvr(\nu) \Big[ \Phi_q (\nu) +\left(\frac{2r_{\vg,0}}{\hvr(\nu)} - 1\right) \Phi'_q(\nu) \Big]
\end{align}\changee
which results in the filter power spectrum \eqref{eq:filteradj}
\begin{align}
  \Betrag{\hvgtDzweig(\nu)}^2 =\left(1+\left(\frac{2r_{\vg,0}}{\hvr(\nu)} - 1\right) \cdot \frac{\Phi'_q
  (\nu)}{\Phi_q(\nu)}\right)^{-1}\label{eq:filteradjcon}.
\end{align}
But since we fixed $\Delta=2$ and $q$ we can calculate $\Phi_q,\Phi'_q$ and $\hqoD$. This provides a separation of the
filter power spectrum $\hvr(\nu)=\Betrag{\hvgL(\nu)}^2$ and the orthogonalization. Unfortunately, this does not yield
linear constraints for $\vr$.\\
Finally, note that the  time-shifts and hence the filter clock rates have to be chosen such that an overlap of the
basic pulse occurs.  Otherwise a frequency  shaping is  not possible.

\changes Case 3) assumes already a shift-orthogonality. We only have to ensure that the
spectral optimization filter $\vgL$ preserves the orthogonality. This results in an extra orthogonality constrain for the
filter $\vgL$, which can be easily incorporated in the SDP problem of \secref{sec:pulseshapedesign}, see Davidson
et. al in \cite{DLW00}. \changee

Summarizing, the discussion above shows that joint optimization and orthogonalization is a complicated problem and only
in specific situations a closed-form solution seems to be possible. 

%%%%%%%%%%%%%%%%%%%%%%%%%%%%%%%%%%%%%%%%%%%%%%%%%%%%%%%%%%%%%%%%%%%%%%%%%%%%%%%%%%%%%%%%%%%%%%%%%%%%%%%%%
\subsection{Compactly Supported Orthogonal Generators}\label{sec:cso}
%%%%%%%%%%%%%%%%%%%%%%%%%%%%%%%%%%%%%%%%%%%%%%%%%%%%%%%%%%%%%%%%%%%%%%%%%%%%%%%%%%%%%%%%%%%%%%%%%%%%%%%%%
%
For PPM transmission  a time-limited shift-orthogonal pulse is necessary to guarantee an ISI free modulation in
a finite time. Such a PPM pulse is a compactly supported orthogonal (CSO) generator. In PPM this is simply realized  by
avoiding the overlap of translates. 

To apply our OOPPM design it is hence necessary to guarantee a compact support of the Löwdin generator $\po$ given in
\thmref{th:lozak}. In this section we will therefore investigate the support properties of orthogonal generators. The
existence of a CSO generator (with overlap) was already shown by Daubechies in \cite{Dau90}.
Unfortunately, she could not derive a closed-form for such an CSO generator. 
Moreover, to obtain a realizable construction of a CSO generator this construction has to be performed in a finite time.
So our Löwdin construction should be obtained by a FIR filter.

PSI spaces of compactly supported (CS) generators, were characterized in detail by de Boor et al.  in \cite{BDR94b} and
called \emph{local PSI spaces}. If the generator is also stable, as in \thmref{th:lozak}, then there exists a sequence
$\vc\in\ell^2$ such that $p\convdis \vc$ is an CSO generator. Moreover, any CSO generator is of this form.  To
investigate compactness, de Boor et.al.  introduced the concept of \emph{linear independent} shifts for CS
generators.  The linear independence property of a CS generator $p$ is equivalent by \cite[Res. 2.24]{BDR94b}  to
\begin{align}
\{\Laplacep(z-n)\}_{n\in\Z} \not= 0 \quad \text{ for all}\quad z\in\C \label{eq:linind}
\end{align} 
where $\Laplacep$ denotes the Laplace transform of $p$. This means $\Laplacep$ do not have periodic zero points.  Note
that this definition of independence is stronger than finitely independence, see definition in \cite{BDR94b}. If we
additionally demand  linear independence of $p$ in our  Theorem \ref{th:lozak}, then this CS generator
is unique up to shifts and scalar multiplies.  Furthermore, a negative result is shown in \cite{BDR94b}, which excludes
the existence of a CSO generator if $p$ itself is not already orthogonal. But if $p$ is already orthogonal, then $p$ is
unique up to shifts and scalar multiplies and then the Löwdin construction becomes a scaled identity (normalizing of
$p$). The statement is the following: 

%
%%%%%%%%%%% THEOREM: No filter can construct an compactly supported orthogonal generator %%%%%%%%%%%%%%%%%%%
\begin{thm}[Th.~2.29 in \cite{BDR94b}]\label{th:negative}
  Let $p\in L^2$ be a linear independent generator for $\PSI(p)$ which is not orthogonal, then there does not exists a
  compactly supported orthogonal generator $\po$ for $\PSI(p)$, i.e. there exists no filter $\vh\in \ell^2$ such that
  $\po =p \convdis \vh$. \end{thm}
If $p$ is a linear independent generator which is not orthogonal, then the Löwdin generator $\po$
has not compact support. We  extend this together with the existence and uniqueness of a linear independent generator for
a local PSI space $\PSI(p)$:
\begin{corrolary}
  Let $p\in L^2$ be compactly supported. If there exists a compactly supported orthogonal generator $\po$ for $\PSI(p)$,
  then it is unique.
\end{corrolary}
\begin{proof}
  Any CSO generator $\po\in L^2$ is a linear independent generator by
  \cite[Prop.~2.25(c)]{BDR94b}. Since the linear independent generator is unique by \cite[Th.~2.28(b)]{BDR94b}, the
  CSO generator is as well.
\end{proof}
\begin{bemerkung}
In any case there exists an orthogonal generator for a local PSI space. For a stable CS generator $p$ our Theorem
\ref{th:lozak} gives an explicit construction and approximation for an orthogonal generator by an IIR filtering of $p$.
If the Löwdin generator is not CS, it is the unique  orthogonal generator with the minimal $L^2$-distance to the
original stable CS generator by Theorem \ref{th:unique}.  So far it is not clear whether there exists a IIR filter
$\vc\in\ell^2$ which generates a CSO generator from a stable CS generator or not. What we can say is that if the inverse
square-root of the Gram matrix is banded, then the rows corresponds to FIR filters which produce CSO generators, since the
semi-discrete convolution reduces to a finite linear combination of CS generators. So this is a sufficient condition for
the Löwdin generator to be CSO, but not a necessary one.
\end{bemerkung}
%clearpage
%
%%%%%%%%%%%%%%%%%%%%%%%%%%%%%%%%%%%%%%%%%%%%%%%%%%%%%%%%%%%%%%%%%%%%%%%%%%%%%%%%%%%%%%%%%%%%%%%%%%%%%%%%%
\changes\section{Application in UWB Impulse Radio Systems}\label{sec:app}\changee
%%%%%%%%%%%%%%%%%%%%%%%%%%%%%%%%%%%%%%%%%%%%%%%%%%%%%%%%%%%%%%%%%%%%%%%%%%%%%%%%%%%%%%%%%%%%%%%%%%%%%%%%%

%

  Here we give some exemplary applications  of our filter designs developed in \secref{sec:pulseshapedesign} and
  \secref{sec:orthogonalization} for UWB-IR.
\begin{paragraph}{FIR filter realized by a distributed transversal filter}
  The FIR filter is completely realized in an analog fashion. It consists of time-delay lines and multiplication of the
  input with the filter constants. Note also that these filter values are real-valued. An application \changes in UWB
  radios \changee was already considered in \cite{ZZMW09}.
\end{paragraph}

\begin{paragraph}{Transmitter and receiver designs}
  Our channel model is an AWGN channel, i.e. the received signal $r(t)$ 
  is the transmitted UWB signal $u(t)$ given in \eqref{eq:APAMPPM} \changes by adding \changee white Gaussian noise.
  For simplicity of the discussion we omitted the time-hopping sequence \changes $\{c_n\}$ \changee in \eqref{eq:uwbsignal}.
\begin{figure}
  \centering
  \includegraphics[scale=0.42]{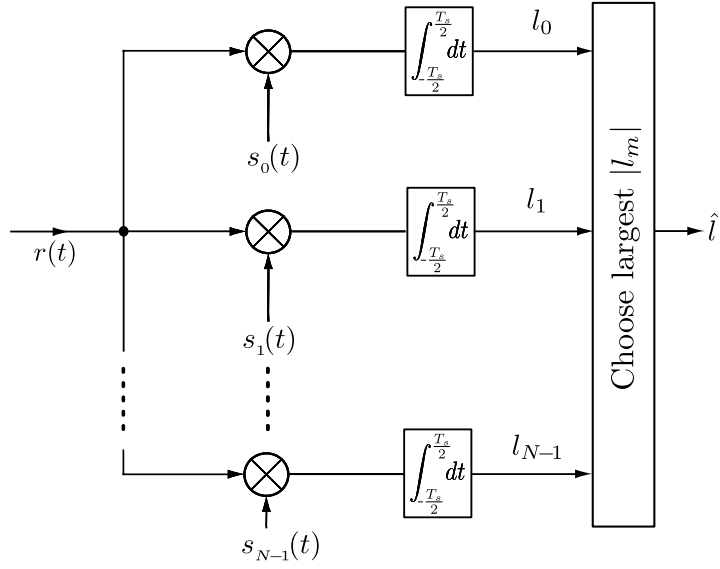}
  \caption{Optimal Receiver for $N$-ary orthogonal modulation with scaled Löwdin pulses $s_m := 
  \sqrt{{\mathcal E}_n}{p}^{T,\circ,M}_{m}$.}
  \label{fig:orlomulti}
\end{figure}
\noi We propose now three $N-ary$ waveform modulations for our pulse design.  \changes Since our proposed modulations are
linear and performed in the baseband, the signals (pulses) are \changee real-valued.\footnote{\changes Note that our
proposed pulse design can be also used for a complex modulations (carrier based modulation), e.g. for OFDM or
FSK.\changee}
\begin{enumerate}[(a)] 
  \item A pulse shape modulation (PSM) with the Löwdin pulses $\{\pToM_m\}_{m=-M}^M$, which corresponds to a $N$-ary
    orthogonal waveform modulation. \changes The $n$'th message $m$ is transmitted as the signal
    $u(t)= \sqrt{\En}\  a_n \pToM_{m} (t-nT_s)=a_n s_m (t-nT_s)$. \changee The receiver is realized by $N$
    correlators using the Löwdin pulses as templates. \changes From the
    correlators output $l_m$  the absolute value is taken due \changee
    to the random amplitude flip by \changes $a_n$, see \figref{fig:orlomulti}.\changee
    \label{en:lo} 
  \item The centered ALO and LO pulse $\tpToM:=\tpToM_0$ resp. $\pToM:=\pToM_0$ for an OPPM design are in fact a
    non-orthogonal modulation scheme with a matched filter at the receiver, \label{en:oppm} see \figref{fig:matched}.
    \changes The $n$'th message $m$ is then transmitted as ${u}(t) = \sqrt{\En}\  a_n \tpToM (t-nT_s - mT)$ in the
    ALO-OPPM
    design. The matched filter output $y(t)=\int_{-\infty}^{nT_s+T_s/2} r(\tau)\tpToM (\tau-t) d\tau$  is sampled for
    the $n$'th message at $y(nT_s+mT)$. (for LO-OPPM  use $\pToM$) \changee
  \item The limiting OOPPM design with the Löwdin pulse $\po$ \label{en:nyquist} is not practically feasible, since we
    have to use an IIR filter.  Hence we only refer to this setup as the theoretical limit. \changes The
    transmitted signal would be $u (t) =  \sqrt{\En}\  a_n \pTo (t-nT_s - mT)$ with the matched filter $h(t)=\pTo (t)$.
    Note that the receiver in \figref{fig:matched} also has to integrate over the whole time due to the unlimited
    support, which would produce an infinity delay in the decoding process.\changee
\end{enumerate}
\end{paragraph}

\begin{figure}
  \centering
  \includegraphics[scale=0.42,angle=0]{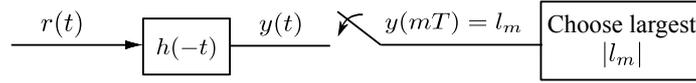}
  \caption{Matched filter receiver for an ALO and LO-OPPM scheme.}
  \label{fig:matched}
\end{figure}

  \begin{paragraph}{Scaling with respect to the FCC mask}\label{scaling}
    The operations $B^M$ and $\tilde{B}^M$ generate pulses which are normalized in energy but do not respect anymore 
    the FCC mask. So we have to find for the $m$th pulse its maximal scaling factor $\alpha_m >0$ s.t.
    \begin{align}
      \Betrag{\alpha_m \cdot\hpoMm(\nu)}^2 \leq \SFCC(\nu)
    \end{align}
    is still valid for any $\nu\in F$. This problem is solved by 
    \begin{align}
      \alpha^*_m= \Norm{\frac{\Betrag{\hpoMm}^2}{\SFCC}}_{L^\infty
      ([0,F])}^{-\frac{1}{2}}
      =\Norm{\frac{\sqrt{\SFCC}}{\hpoMm}}_{L^\infty([0,F])}.\label{eq:alphastern}
    \end{align}
    For the scaled Löwdin pulses we can easily obtain the following upper bound for the NESP value \eqref{eq:nesp} 
    \begin{align}
      \nesp(\alpha^*_m\poMm) = \frac{\int_{F_p} \Betrag{\alpha^*_m \hpoMm}^2}{\int_{F_p} \SFCC }  
      \leq  \frac{\Norm{\hpoMm}^2_{L^2}}{\En_{F_p} }  
             \cdot\Norm{\frac{\sqrt{\SFCC}}{\hpoMm}}^2_{L^\infty([0,F]}
      =  \frac{1}{\En_{F_p} }  
        \cdot\Norm{\frac{\sqrt{\SFCC}}{\hpoMm}}^2_{L^\infty([0,F]}
      =  \frac{1}{\En_{F_p} }  
      \cdot\Norm{\frac{\SFCC}{\Betrag{\hpoMm}^2}}_{L^\infty([0,F]} \label{eq:l2linfty}
    \end{align}
    where we denoted with $\En_{F_p}=\int_{F_p} \SFCC$ the allowed energy of the FCC mask in the passband $F_p$. 
    Thus, the maximization of the symbol energy \changes $\En$ \changee under the FCC
    mask is the maximization of $\alpha^*_m$ in \eqref{eq:alphastern}, i.e. a maximization of the $L^\infty$-norm in the
    frequency domain.
\end{paragraph}
\changes
\begin{bemerkung} We want to emphasize at this section, that the FCC spectral optimization is rather an optimization of 
  the power in an allowable mask given by the FCC, than an optimization of the spectral efficiency of the signal.
  Although the PSWF's have the best energy concentration in $[-W,W]$ among all time-limited finite energy signals
  \cite{PCWD03}, they are not achieving the best possible nesp value \cite{WTDG06}. The reason is, that  spectral efficiency with
  respect to UWB is to utilize as much power from the UWB window, framed by the FCC mask $\SFCC$, as possible. So the
  price of power utilization in $[-W,W]$ under the UWB peak power limit $\SFCC$, is a loss of energy concentration and
  hence a loss of spectral efficiency compared to the PSWF's. 
\end{bemerkung}
\changee

%%%%%%%%%%%%%%%%%%%%%%%%%%%%%%%%%%%%%%%%%%%%%%%%%%%%%%%%%%%%%%%%%%%%%%%%%%%%%%%%%%%%%%%%%%%%%%%%%%%%%%%%%
\subsection{Performance of the Proposed Designs}\label{sec:performance}
%%%%%%%%%%%%%%%%%%%%%%%%%%%%%%%%%%%%%%%%%%%%%%%%%%%%%%%%%%%%%%%%%%%%%%%%%%%%%%%%%%%%%%%%%%%%%%%%%%%%%%%%%
  For a given transmission design, consisting  of a modulation scheme and a receiver, the average bit error probability
  $P_e$ over $E_b/N_0$ is usually considered as the performance criterion. We consider real-valued signals in the
  baseband with finite symbol duration $T_s$. The optimal receiver for a non-orthogonal
  $N$-ary waveform transmission is the correlation receiver with $M$ correlators, see \figref{fig:orlomulti} with 
  maximum likelihood decision.
  \begin{paragraph}{$N$-ary orthogonal PSM for scheme \eqref{en:lo} above}  
    The average (symbol) error probability for $N$-ary orthogonal pulses with equal energy $\En$ can be upper bounded by
    \cite{Tre68}
    \begin{align}
      P_e \leq (N-1) \erfc\left(\sqrt{\frac{\En}{N_0}}\right) \label{eq:mpBERupper}
    \end{align}
    Note, that this error probability is the same as for an orthogonal PPM modulation \eqref{eq:performpar}.  To obtain
    equal energy symbols and FCC compliance we have to scale each Löwdin pulse with $\sqrt{\En}= \alpha^*  =
    \min_m\{\alpha^*_m\}$.
  \end{paragraph}
  \begin{paragraph}{$N$-ary overlapping PPM for scheme \eqref{en:oppm} above}
  Here we can substitute the $N$ correlations by one matched filter $h=\pToM$ resp. $\tilde{h}=\tpToM$ and obtain the
  statistics $|l_m|$ by sampling the output.  The average error probability per symbol $P_e$  given exactly in
  \cite[Prob. 4.2.11]{Tre68} for equal energy signals and can be computed numerically. The energy is given by
  $\sqrt{\En}=\alpha^*=\alpha^*_0$ and $\sqrt{\tilde{\En}}=\tilde{\alpha}^*_0$ calculated in \eqref{eq:alphastern} for
  $\pToM$ resp. $\tpToM$. Upper bounds obtained in
  \cite{Jac67} can be used for the ALO resp. LO average  error probability
  \begin{align}
    P_{e} \leq \frac{1}{2} \sum_{j=2}^N \erfc \left(\sqrt{\frac{\En}{2N_0} (1-\rho_{1j})}\right) \quad\text{and}\quad 
    \tilde{P}_{e}
    \leq \frac{1}{2} \sum_{j=2}^N \erfc \left(\sqrt{\frac{\tilde{\En}}{2N_0} (1-\tilde{\rho}_{1j})}\right)
    \label{eq:BERupper}
  \end{align}
  with $\rho_{1j}=\En r_{\pToM} (jT)$ and $\tilde{\rho}_{1j}=\tilde{\En} r_{\pToM}(jT)$, since the symbols are given by
  $s_j = \sqrt{\En}\pToM(\cdot - jT)$ resp. $\tilde{s}_j =\sqrt{\tilde{\En}}\tpToM(\cdot - jT)$ for $j=1,\dots,N$. The error
  probabilities depend on the pulse energy and on the decay of the sampled auto-correlation defined in
  \eqref{eq:autocor}.
\end{paragraph}

%%%%%%%%%%%%%%%%%%%%%%%%%%%%%%%%%%%%%%%%%%%%%%%%%%%%%%%%%%%%%%%%%%%%%%%%%%%%%%%%%%%%%%%%%%%%%%%%%%%%%%%%%
\subsection{Simulation Results}
%%%%%%%%%%%%%%%%%%%%%%%%%%%%%%%%%%%%%%%%%%%%%%%%%%%%%%%%%%%%%%%%%%%%%%%%%%%%%%%%%%%%%%%%%%%%%%%%%%%%%%%%%

The most common basic pulse for an UWB-IR transmission is the Gaussian monocycle: $q(t)\simeq
t\cdot\exp{(-t^2/\sigma^2)}$ where $\sig$ is chosen such that the maximum of $\Betrag{\hq(f)}^2$ is reached at the
center frequency $f_c=6.85$GHz of the passband \cite{WJT10}.  Since we need compact support and continuity for our
construction, we mask $q$ with a unit triangle window $\Lam$ instead of a simple truncation.  Also any other
continuous window function which goes continuously to zero (e.g. the Hann window) can be used, as long as the lower
Riesz bound $A>0$ can be ensured, see \thmref{th:lozak}. We have used an algorithm in \cite{WJ10} to compute $A$ and $B$
numerically.
Note that for any continuous compactly supported function we have a finite upper Riesz bound $B$, see
\cite[Th.2.1]{JM91}.  The width (window length)  is chosen to  $T_q=T_{\Lambda}=N_q T_0\approx0.21428$ns, such
that at least $99.99\%$ of the energy of $q$ is contained in the window $[-T_q/2,T_q/2]$, see \figref{fig:timeOptimal}.
We express all time instants as  integer multiples of $T_0$.  Also, in \figref{fig:timeOptimal} we plot the optimal
pulse obtained by a FIR filter of order $L=25$ which results in a time-duration $T_p=30T_0=5T_q$ of $p$.
\begin{figure}
  \centering
  \includegraphics[scale=0.45]{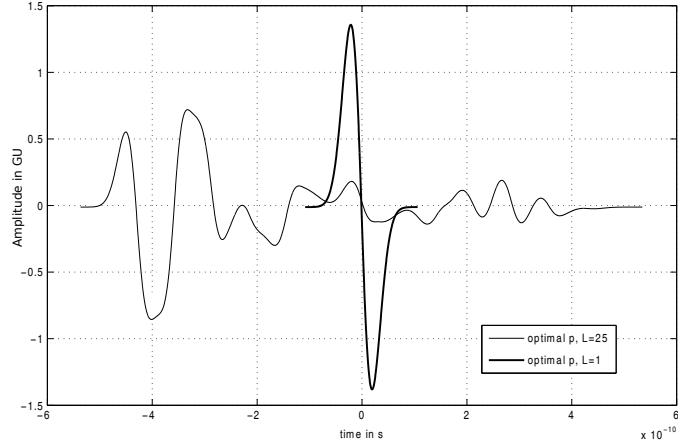}
  \caption{Optimal pulse $p$ for $L\!=\!25$ and basic pulse $q=p$ for $L\!=\!1$ in time-domain in generic units (GU).}
  \label{fig:timeOptimal}
\end{figure}
In our simulation we choose $T_q:=6 T_0 =N_q T_0$ and $L=25$ as the filter order of the FCC-optimization. Hence, the
optimized pulses have a total time duration of $T_p=(L-1)T_0+T_q=30T_0=N_p T_0$.

The Riesz condition \eqref{eq:rieszcondition} has been already verified in \cite{WJ10} for this particular setup.
Theorem \ref{th:lozak} uses the normalization $T'=1$.  Translating between different support lengths $T'_p=K$ is done
by setting $t:=t' T_p/K$. Now the support of $p(t')$ is $[-K/2,K/2]$ with fixed $T'=1$.  To obtain good
shift-orthogonality, we have to choose $M>K$. This we control with an integer multiple  $m=2$, i.e.  $M=mK=2K$.  The
support length $T_s$ of all the LO (ALO) pulses is then given as 
\begin{align}
  \begin{split}
  T_{\tpToM}&=(N-1)T + T_p = (2mK)T_p/K + N_p T_0 \\
  &=T_{\pToM}=(2m+1)N_p T_0=150T_0.
\end{split}
\end{align}
Now the time slot $[-T_{\pToM}/2,T_{\pToM}/2]$ exactly contains $N$ mutually orthogonal pulses $\{\pToM_m\}$, i.e. $N$
orthogonal symbols with symbol duration $T_s=T_{\pToM}$ having all the same energy and respecting strictly the FCC mask.
This is a $N$-ary orthogonal signal design, which requires high complexity at receiver and transmitter, since we need a
filter bank of $N$ different filters.\\
Our proposal goes one step further. If we only consider one filter, which generates at the output the centered Löwdin
\changes orthonormal \changee pulse $\pToM$, we can use this as a approximated \changes square-root \changee Nyquist
pulse with a PPM shift of $T=T_p/K$ to enable $N$-ary OPPM transmission by obtaining almost orthogonality.

Advantages of the proposed design are:
a low complexity at transmitter and receiver, 
a combining of $\vg$ and $\vh$ into a single filter operating with clock rate $1/T_0$ and $q$ as input
if $T=T_p/(T_0 K) \in\N$
, a signal processing ''On the fly'' and finally a much higher bit rate compared to a binary-PPM.
The only precondition for all this, is a perfect synchronisation between transmitter and receiver.
In fact, we  have to sample equidistantly at rate of $1/T$. The output of the 
matched filter $h(-t)=\pToM(t)$  is given by
\begin{align}
  y(t) = \int_{-\infty}^{\infty} r(\tau) \pToM (\tau-t) d\tau
\end{align}
and recovers the $m$th symbol. The statistic $l_m=y(mT)$ is the correlation of the received signal with the symbol
$s_m$, see \figref{fig:matched}.
Note that the shifts have support in the window $[-1.5T_{\pToM},1.5T_{\pToM}]$, but are
 almost orthogonal outside the symbol window
$[-T_{\pToM}/2,T_{\pToM}/2]$ due to the compactness and approximate shift-orthogonal
character of the Löwdin pulse $\pToM$.

In \figref{fig:timeFIRout} the centered orthogonal pulses for $T=5T_0=5T_q/6$ match almost everywhere the original masked
Gaussian monocycle, since the translates are almost non-overlapping, hence they are already almost orthogonal. For
$T=T_0=T_q/6$ the overlap results in a distortion of the centered orthogonal pulses, where the ALO pulses have 
high energy concentration at the boundary (circulant extension of the Gram Matrix).
\begin{figure}
  \centering
  \includegraphics[scale=0.43]{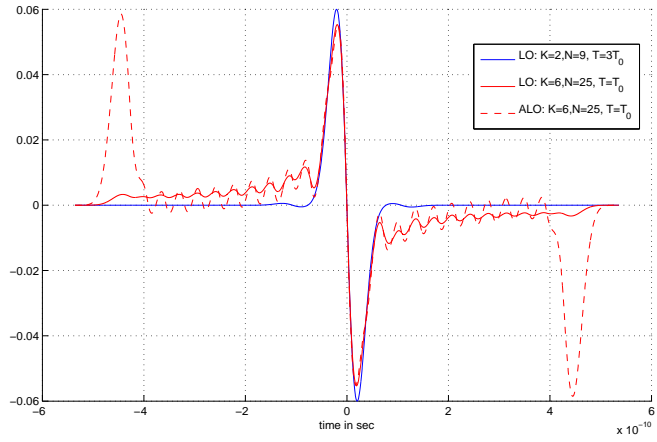}
  \caption{ALO and LO pulses for Gaussian monocycle in time, $L=1$ and $M=2K$.}
  \label{fig:timeFIRout}
\end{figure}
\begin{figure}
  \centering
  \includegraphics[scale=0.57]{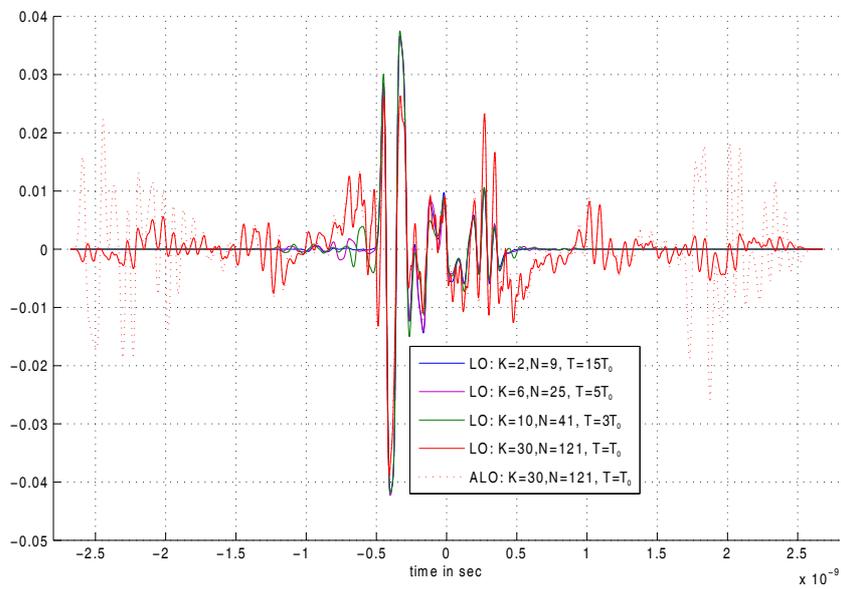}
  \caption{Orthogonal pulses $p^{T,\circ,M}$ with $L=25, M=2K$ and various $T=T_p/K$.}
  \label{fig:timeLO}
\end{figure}

\begin{figure}
  \centering
  \includegraphics[scale=0.48]{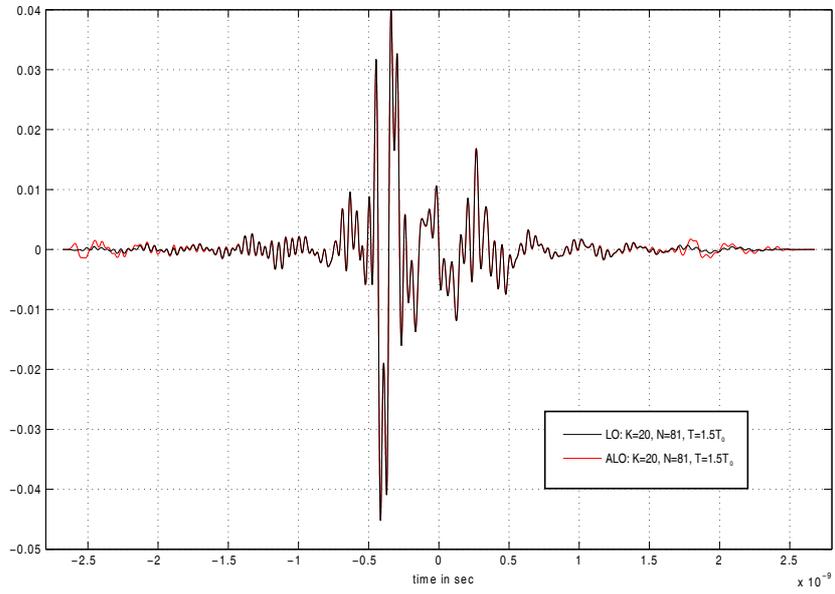}
  \caption{Centered ALO and LO pulse for $T=1.5T_0$ with $L=25, M=2K=40$.}
  \label{fig:timeALOandLO}
\end{figure}
\begin{figure}
  \centering
  \includegraphics[scale=0.57]{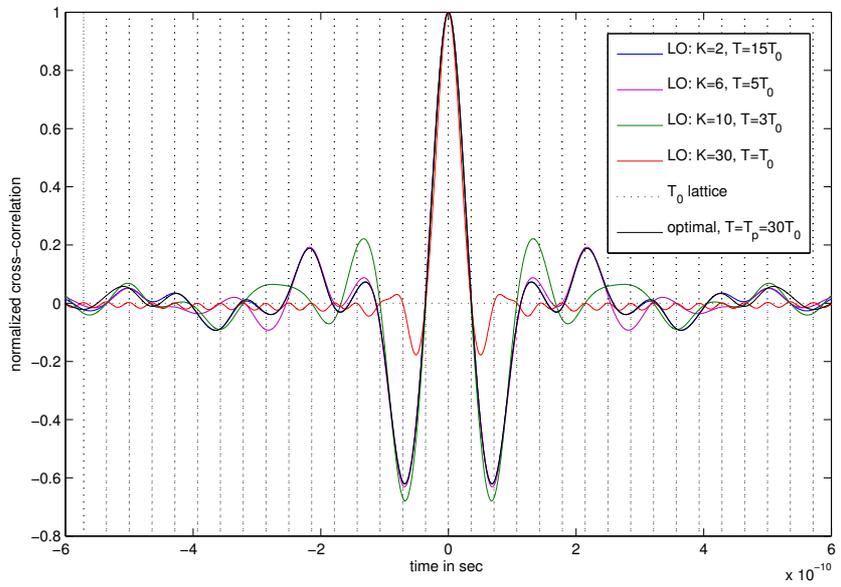}
  \caption{Auto-correlation of pulses for $L=25,T=T_p/K$.}\label{fig:pulses52}
\end{figure}
%

\begin{comment}
%
Let us scale the time $t'=t T_p/M$ such that the support is $[-M/2,M/2]$ and the time shift $T$ is one. 
%
If we fix $M$ and $p$, then $\tpoKo$ is the $K$--approximation to the limit pulse $\po:=\poo$ which is an orthogonal
generator for $\SI(p)$, hence we call them approximated L{\"o}wdin orthogonal (ALO).  The question is, how well
$\{\tpoKk\}$ approximates a shift-orthogonal set for finite $K$?

For this we compute with MATLAB the ALO pulses $\tpoKo$ in by using the DFT and plot the result in
\figref{fig:timeOptimal} for the parameters $M=4$ and $K=12$ together with the initial pulse $p$ and the L{\"o}wdin
orthogonal (LO) pulse $\poKo$ in \eqref{eq:co}.  It can be observed that $\tpoKo$ and $\poKo$ are quite similar but have
increased support as compared to $p$.  Due to the circulant construction $\tpoKo$ has slight more concentration at the
boundaries which vanishes with increasing $K$. 

We discover, that the time extension due to the FIR optimization filter $\vg$ can be compensated by the
orthogonalization filter $\vh$, which allows overlapping orthogonal PPM transmission at the same pulse repetition time
as with the non shaped Gaussian monocycle with non-overlapping orthogonal PPM.

%
We choose in the first plots $M=4$ and $K=8$. This results in $2M-1$ shifted pulses which overlap with the generator
pulse $p$. See .
\end{comment}

 In \figref{fig:timeLO}\label{remark} the pulse shapes in time for the centered Löwdin orthogonal pulses are plotted. The ALO
pulses are matching the LO pulses for $T>2T_0$ almost perfectly such that we did not plot them, since the resolution of
the plot is to small to see any mismatch. Only for the critical shift $T=T_0$ a visible distortion is obtained at the
boundary. In the next \figref{fig:timeALOandLO} we plotted therefore the ALO and LO pulse for $T=1.5T_0$ to show that
the ALO pulses indeed converge very fast to the LO pulse if $T\gg T_0$. The reason is that for small time shifts of $p$
the Riesz bounds and so the clustering behaviour of $\GramM$ and $\GramMinv$ decreases.  Hence the approximation quality
of $\GramMinv$ with $\tGramMinv$ decreases, which results in a shape difference.  

To study the shift-orthogonal character of the ALO and LO pulses for various $T$, we have plotted the auto-correlations
in \figref{fig:pulses52}. As can be seen, the samples $r_{\pToM}(mT)=\rho_{1m}\approx \del_{m0}$, i.e. they vanish at
almost each sample point except in the origin. 
The NESP performance for various values of $T$ is shown in \figref{fig:t2t1}. Approaching $T=T_0$
cancels the FIR prefilter optimization of $\vg$, i.e. the spectrum becomes flat.
%
%The problem here is that a certain amount of the energy of the ALO pulses is concentrated on the boundary of the time
%support. This results in reduced cross correlation quality. Nevertheless, the shapes in the center of the time region
%are well matching.  The question is, what are the gains of increasing $K$? It should compensate the energy
%deviation. In \figref{fig:pulses52} this was done for $K=52$.
%The energy concentration on the boundary is much reduced if we increase the ortho-shift or increase the mulitplier $m$.
 %Anyway, the last
%ALO pulse has still all its energy concentrated at the boundary. So the effect of the ALO pulse approximation is to
%shift the effect of the edges in the Strang-Gram matrix to the edges in the time area, which means in the limit to
%infinity.  
%Simultaneously this results also in a better energy concentration of the ALO pulses in the centered region.
%
Finally, in \figref{fig:OrthofeatOpti} the gain of our orthogonalization strategy can be seen. 
In both cases, the $N$-ary OOPPM design  transmit at \changes an uncoded bit \changee rate of

\begin{align}
  R_b(K) = \frac{\log(N)}{T_s} = \frac{\log(4K+1)}{150T_0}.\label{eq:bdr}
\end{align}
\figref{fig:OrthofeatOpti} shows the NESP value \changes $\eta$ over the  transmit rate $R_b$ for $L=25$ and $m=2$.
Decreasing $T=T_p/K$ results in more overlap, which increases the number of symbols $N$ in $T_s$ and hence $R_b$, but
only \changee slightly decreases $\eta$, see \changes \figref{fig:t2t1},\ref{fig:OrthofeatOpti}.\changee

Summarizing, a  triplication of the \changes transmit rate \changee  from $0.18$Gbit/s to $0.6$Gbit/s is possible
without loosing much \changes signal power $\En/T_s$ \changee. Let us note the fact that the \changes transmit \changee
bit data rate is an uncoded rate \changes and is not an  achievable rate. 
For an analysis of achievable data rates by deriving the mutual information of the system see the work of Ramirez
et. al in \cite{RM07} and Güney et. al in \cite{GDA10}.
\changee Obviously, the unshaped Gaussian
monocycle then yields the highest \changes transmit \changee rate, since \eqref{eq:bdr} behaves logarithmically in the
number $N$ of symbols, as seen in \figref{fig:OrthofeatOpti}. But this has practically zero SNR when respecting the FCC
regulation \changes and results in a high error rate \eqref{eq:mpBERupper},\eqref{eq:BERupper}.  \changee  On the other
hand, a longer symbol duration, allows in \eqref{eq:psd}  a higher energy $\En$ and hence a lower error rate in
\eqref{eq:performpar}.

\changes\emph{Hence, the decreasing of the transmit rate due to the increased symbol duration used for FCC optimal filtering of
the masked Gaussian monocycle can be more than compensated by the proposed OOPPM technique.}\\
\changee
\begin{figure}
  \centering
  \includegraphics[scale=0.53]{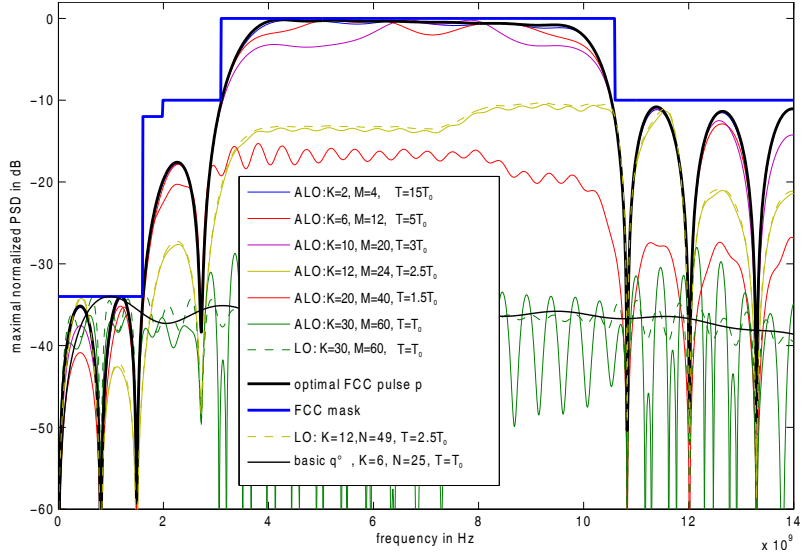}
  \caption{PSD of pulses for $L=25,T=T_p/K$.}\label{fig:t2t1}
\end{figure}
\begin{figure}
  \centering \includegraphics[width=0.70\textwidth]{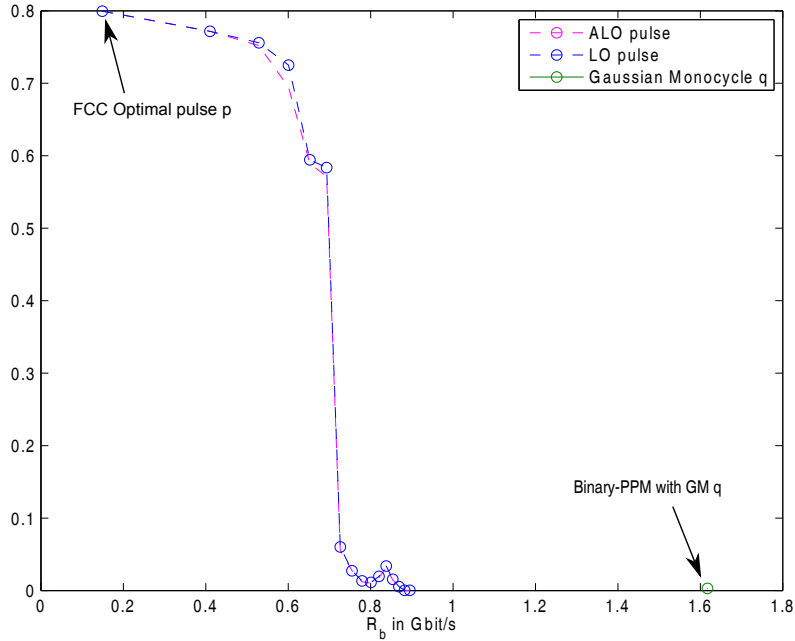} 
  \caption{Orthogonalization results of
  the fixed FCC optimized pulse $p$ for $L=25$ in $T_s=150T_0$ over various $T$ resulting in
  various  transmit rates $R_b$.}
  \label{fig:OrthofeatOpti}
\end{figure}

%%%%%%%%%%%%%%%%%%%%%%%%%%%%%%%%%%%%%%%%%%%%%%%%%%%%%%%%%%%%%%%%%%%%%%%%%%%%%%%%%%%%%%%%%%%%%%%%%%%%%%%%%
\section{Conclusion}
%%%%%%%%%%%%%%%%%%%%%%%%%%%%%%%%%%%%%%%%%%%%%%%%%%%%%%%%%%%%%%%%%%%%%%%%%%%%%%%%%%%%%%%%%%%%%%%%%%%%%%%%%

We have proposed a new pulse design method for UWB-IR which provides high spectral efficiency under FCC constraints and
allowing a $N$-ary OPPM transmission with finite transmission and receiving time by keeping almost orthogonality. In
fact, the correlation parameters can be keep below the noise level by using small time-shifts $T<T_p$. As a
result, this provides much higher data rates as compared to BPSK or BPPM. Furthermore, our pulse design provides a N-ary
orthogonal PSM transmission by getting a lower bit error rate at the price of a higher complexity. 
%Also our modulation scheme is robust against sign flips of the signals, which could occur due to reflection in
%multipath scenarios. 

Simultaneous orthogonalization and spectral frequency shaping is a challenging and hard problem.
We believe that for certain shifts being integer multiples of $T_0$, a numerical solver might be helpfully to directly
solve the combined problem as discusses in \secref{sec:interdependence}.

We highlight the broad application of the OOPPM design, not only being limited to UWB systems but rather 
\changes applicable to a pulse shaped communication system under a \changee local frequency constraint in general.
\section*{Acknowledgements}
Thanks to Holger Boche  for helpful discussions.
This work was partly supported by the Deutsche Forschungsgemeinschaft
(DFG) grants Bo 1734/13-1, WI 1044/25-1 (UKoLoS) and JU 2795/1-1.

\appendix
\section{}\label{app:meyer}
\begin{proof}[Proof of \lemref{le:meyer}]
  Let $a,b\in\R$ with $a+b=-1$. Then $S^a$ and $S^b$, defined by the spectral theorem, are also positive and
  self-adjoint on $\Hil$.  Moreover  for each $f$ we have the following unique representation $f= \sum_k c_k p_k$ with
  $\vc\in l^2(\Z)$ due to the Riesz basis property. For $f=p_l$ in \eqref{eq:frameinv} we get
  \begin{align}
    p_l &= \sum_k \SPH{S^{-1} p_l}{p_k} = \sum_k \SPH{S^{a} p_l}{S^b p_k} p_k 
  \intertext{since $S^{a} p_l , S^b p_k \in\Hil $ there exist unique sequences $\vc_l,\vd_k$ s.t. $S^{a}p_l =
    \sum_\alpha c_{l \alpha} p_\alpha, S^{b} p_k = \sum_\beta d_{k \beta} p_\beta$. Hence we get}
     p_l  &= \sum_k \left(\sum_\alpha c_{l \alpha} p_\alpha, \sum_\beta d_{k \beta} p_\beta\right) p_k\\
     &= \sum_k \sum_{\alpha, \beta} c_{l\alpha} \bar{d}_{k \beta} \SPH{p_\alpha}{p_\beta} p_k\\
      &= \sum_k \sum_{\alpha, \beta} [\mC]_{l\alpha} [\Gram]_{\alpha \beta} [\mD^*]_{\beta k} p_k 
          = \sum_k [\mC \Gram \mD^*]_{lk} p_k
  \end{align}
  where $c_{l\alpha}$ and $d_{\beta k}$ are the coefficients of the biinfinite matrices $\mC$ resp. $\mD$.  Since for
  each $l\in\Z$ we have $\sum_k \delta_{lk} p_k  = p_l$ and $\{[\mC \Gram \mD^* ]_{lk} \}, \{\del_{lk}\}\in \ell^2$, we
  get $0=\sum_k ([\mC \Gram \mD^*]_{lk} -\del_{lk} )p_k$ for each $l\in\Z$.  So by \cite[Th.6.1.1(vii)]{Chr03} we can
  conclude that $[\mC\Gram \mD^*]_{lk}=\del_{lk}$ for all $l,k\in\Z$ and get 
  \begin{align}
      \mC \Gram \mD^* = \eins \LRA \ \Gram = \mC^{-1} (\mD^*)^{-1}  \LRA \  \Gram^{-1}  = \mD^* \mC.
  \end{align}
  Obviously $\mD$ and $\mC$ are not an unique decomposition of $S^{-1}$, since $a$ and $b$ are not.  If $a=b=
  -\frac{1}{2}$, we have $\mD^*=\mC^*$ and hence $\Graminv=\mC=\mD$. This establishes \eqref{eq:framegram} in an
  $L^2$-sense. 
\end{proof}

\section{}\label{app:proofunique}
\begin{proof}[Proof of \thmref{th:unique}]
 Let us first note that $\SI(p)$ is a regular SI space since $p$ is a stable generator. This has as
 consequence that frames are Riesz bases for $\SI(p)$ \cite[Th. 2.2.7 (e)]{RS94}.
 So any element $f\in\SI(p)=p \convdis \vc$ is uniquely determined by an $\ell^2$ sequence $\vc$.  By the Riesz--Fischer
 Theorem this sequence $\vc$ defines by its Fourier series a unique $L^2([0,1])$-function $\tau=\hvc$. Hence,
 the Fourier transform of any $f\in\SI(p)$ is represented uniquely by $\tau$ as $ \hat{f} = \tau \hp$, see also
 \cite[Th.2.10(d)]{BDR94b}.  On the other hand $f$ is an orthonormal generator if and only if $\Phi_{f} =1$ a.e..
  By using the periodicity of $\tau$ we get
 \begin{align}
    \Phi_{f}&= \sum_k \Betrag{\hp(\cdot -k)}^2 \Betrag{\tau(\cdot -k)}^2\\
    &= \Betrag{\tau}^2 \sum_k \Betrag{\hp(\cdot -k )}^2 = \Betrag{\tau}^2 \cdot \Phi_p = 1.\label{eq:orthocond}
  \end{align}
 Thus, we have $\Betrag{\tau}= 1/\sqrt{\Phi_p}$ almost everywhere. Let us set $\tilde{\tau}:=1/\sqrt{\Phi_p}$ a.e. and a
 complex periodic phase function $\phi:= e^{i \alpha(\cdot)} : \R \to \set{z\in\C}{\Betrag{z}=1}$ with $\alpha:\R \to
 [0,2\pi]$ $1$-periodic and measurable.  Then any function $\tau\in L^2([0,1])$ which satisfy \eqref{eq:orthocond} a.e.
 is given by $\tau=\tilde{\tau} \cdot \phi$ a.e..  The $L^2$-distance is then given by

  \begin{align}
    \Norm{p-f}^2_{L^2} &= \Norm{ \hp -\hat{f}}^2_{L^2} = \Norm{\hp-\tau\hp}^2_{L^2}
         = \Norm{p}^2_{L^2} + \Norm{\tau\hp}^2_{L^2} - \int_\R \tau \Betrag{\hp}^2 - \int_\R \cc{\tau} \Betrag{\hp}^2 \\
    &= 2- \int_\R (\tau + \cc{\tau})\Betrag{\hp}^2 = 2- 2\int_\R \cos (\alpha) \
    \Betrag{\hp}^2 {\Phi_p}^{\!-\frac{1}{2}}\\
    &\geq 2-2 \int_\R {\Betrag{\hp}^2} {\Phi_p}^{\!-\frac{1}{2}} = 2\left(1 - \SPH{ p }{\po}\right) .\label{eq:lomin}
  \end{align}

Since $\Betrag{\hp}^2$ is positive and $\Phi_p$ is bounded and strictly positive a.e. the distance is minimized if and
only if $\alpha(\nu)=0$ a.e. in $\R$,  i.e. if we have equality in \eqref{eq:lomin}.  Hence $\phi=1$ a.e. and so
$\tilde{\tau}=\tau$ a.e., which corresponds hence to the unique orthonormal Löwdin generator $f=\po$ with an
$L^2$-distance to $p$ given in \eqref{eq:lomin}.
\end{proof}
\begin{bemerkung}
Note, that in fact the phase function $\phi$ has no influence on the power spectrum 
$ |\hpo_\phi|^2=\Betrag{\phi \ttau \hp}^2 =\Betrag{\tilde{\tau}\hp}^2= \Betrag{\hpo}^2.$
\end{bemerkung}

%
%--------------%
% Bibliography %
%--------------%

\section*{References}
%\cleardoublepage
%\lhead{}\rhead{}
%\bibliographystyle{model1-num-names}    % für amrfes
\bibliographystyle{elsarticle-num}   % für elsarticle documents
%\bibliographystyle{IEEEtran}   % für elsarticle documents
%\bibliographystyle{amsxport}
% my_abbr =meine abbkürzungen für journals books etc., ieee_abbr.

\end{document}